\documentclass[11pt]{article}

\PassOptionsToPackage{numbers, compress}{natbib}
\usepackage{amsmath,amssymb,amsthm,fullpage,mathrsfs,pgf,tikz,caption,subcaption,mathtools}
\usepackage[ruled,vlined,linesnumbered]{algorithm2e}
\usepackage{natbib}
\usepackage{amsmath,amssymb,amsthm,mathtools}
\usepackage[utf8]{inputenc} 
\usepackage[T1]{fontenc}    
\usepackage{hyperref}       
\usepackage{url}            
\usepackage{booktabs}       
\usepackage{amsfonts}       
\usepackage{nicefrac}       
\usepackage{microtype}      
\usepackage{times}
\usepackage[margin=1in]{geometry}
\usepackage{bbm}
\usepackage{enumitem}
\usepackage{xcolor}
\usepackage{cleveref}
\usepackage{bm}

\newtheorem{theorem}{Theorem}[section]
\newtheorem{corollary}[theorem]{Corollary}
\newtheorem{proposition}[theorem]{Proposition}
\newtheorem{lemma}[theorem]{Lemma}
\newtheorem{definition}[theorem]{Definition}

\newtheorem{observation}[theorem]{Observation}
\newtheorem{claim}[theorem]{Claim}

\newcommand{\floor}[1]{\lfloor {#1} \rfloor}

\newcommand\norm[1]{\left\lVert#1\right\rVert}

\newcommand{\gup}[1]{g_{\uparrow {#1}}}
\newcommand{\gdown}[1]{g_{\downarrow {#1}}}
\newcommand{\fup}[1]{f_{\uparrow {#1}}}
\newcommand{\fdown}[1]{f_{\downarrow {#1}}}
\newcommand{\restrict}[2]{{#1}|_{#2}}
\newcommand{\localize}[2]{{#1}_{#2}}
\newcommand{\lazy}{\text{lazy}}
\newcommand{\stab}{\text{Stab}}
\newcommand{\var}{\text{Var}}
\newcommand{\eps}{\varepsilon}
\newcommand{\E}{\mathbb{E}}
\newcommand{\ip}[1]{\langle #1 \rangle}

\author{%
  Mitali Bafna\thanks{Department of Computer Science, Harvard University, MA 02138. Email: \texttt{mitalibafna@g.harvard.edu}.}
  \and
  Max Hopkins\thanks{Department of Computer Science and Engineering,
  UCSD, CA 92092. Email: \texttt{nmhopkin@eng.ucsd.edu}. Supported by NSF Award DGE-1650112.}
  \and
    Tali Kaufman\thanks{Department of Computer Science, Bar-Ilan University. Email: \texttt{kaufmant@mit.edu}. Supported by ERC and BSF.}
    \and
    Shachar Lovett\thanks{Department of Computer Science and Engineering, UCSD, CA 92092. Email: \texttt{slovett@cs.ucsd.edu}. Supported by NSF Award CCF-1953928.}
}
\title{Hypercontractivity on High Dimensional Expanders: a Local-to-Global Approach for Higher Moments}

\begin{document}

\maketitle
\begin{abstract}
Hypercontractivity is one of the most powerful tools in Boolean function analysis. Originally studied over the discrete hypercube, recent years have seen increasing interest in extensions to settings like the $p$-biased cube, slice, or Grassmannian, where variants of hypercontractivity have found a number of breakthrough applications including the resolution of Khot’s 2-2 Games Conjecture (Khot, Minzer, Safra FOCS 2018). In this work, we develop a new theory of hypercontractivity on high dimensional expanders (HDX), an important class of expanding complexes that has recently seen similarly impressive applications in both coding theory and approximate sampling. Our results lead to a new understanding of the structure of Boolean functions on HDX, including a tight analog of the KKL Theorem and a new characterization of non-expanding sets. 

Unlike previous settings satisfying hypercontractivity, HDX can be asymmetric, sparse, and very far from products, which makes the application of traditional proof techniques challenging. We handle these barriers with the introduction of two new tools of independent interest: a new explicit combinatorial Fourier basis for HDX that behaves well under restriction, and a new local-to-global method for analyzing higher moments. Interestingly, unlike analogous second moment methods that apply equally across all types of expanding complexes, our tools rely inherently on simplicial structure. This suggests a new distinction among high dimensional expanders based upon their behavior beyond the second moment.
\end{abstract}

\newpage
\section{Introduction}
Introduced over 50 years ago today, \textit{hypercontractivity} remains one of the most powerful tools in the analysis of boolean functions. Originally used to prove numerous landmark results on the discrete hypercube such as the KKL Theorem \cite{kahn1988influence} and Majority is Stablest \cite{mossel2005noise}, the study of hypercontractivity has since seen a resurgence on extended domains such as the $p$-biased cube \cite{keevash2019hypercontractivity}, slice \cite{khot2018small}, and Grassmannian \cite{subhash2018pseudorandom}. Fascinatingly, these regimes all share a common thread: while hypercontractivity doesn't hold in general, it is satisfied for certain classes of \textit{pseudorandom functions}. This recently discovered phenomenon has led to a slew of breakthroughs, most famously including the resolution of Khot's 2-2 Games Conjecture \cite{subhash2018pseudorandom}. Unfortunately, the scope of these results is currently restricted, as all known proof techniques rely on product structure or other strong symmetries, and no unifying theory is known to exist.

In this work we take the first substantive step towards solving this issue with the introduction of a new theory of hypercontractivity for the general class of \textit{high dimensional expanders} (HDX). HDX are a family of expanding complexes that have seen an explosion of work in recent years, leading to major breakthroughs across a number of areas including (among others) the recent construction of c3-LTCs \cite{dinur2021}, and efficient approximate sampling for many important systems (e.g. for matroid bases \cite{anari2019log}, independent sets \cite{anari2020spectral}, Ising models \cite{anari2021entropic}, and more). Our results lead to a new understanding of the structure of boolean functions on HDX, including a tight analog of the KKL Theorem, and a characterization of non-expanding sets similar to that used in the proof of 2-2 Games \cite{subhash2018pseudorandom}. 
Proving such results previously seemed out of reach since HDX are very far from products, asymmetric, and can be quite sparse. To handle these challenges, we introduce a new set of tools including a new explicit Fourier decomposition and a local-to-global method for analyzing higher order moments. Interestingly, unlike previous $\ell_2$-based techniques which apply equally across all types of expanding complexes, our methods rely crucially on the underlying HDX structure being \textit{simplicial}. This suggests a new stratification of spectral HDX based upon their behavior \textit{beyond the second moment}.

\subsection{Contributions}
Before jumping into a more detailed breakdown of our results, we start by giving an informal overview of our main contributions within the broader context of classical Fourier analysis and the theory of high dimensional expanders.

\paragraph{Classical Fourier Analysis:} Classical Fourier Analysis on the discrete hypercube focuses on analyzing functions $f: \{0,1\}^n \to \mathbb{R}$ through their \textit{Fourier Expansion}, a decomposition that breaks $f$ into a series of orthogonal ``level functions,'' each corresponding to the projection of $f$ onto a certain eigenspace of the (noisy) hypercube graph.\footnote{More generally, these are the eigenspaces of the Hamming scheme.} At a basic level, a function's Fourier decomposition gives a nice method for understanding its \textit{second moment}, since orthogonality allows one to move between this and the standard basis freely (a result usually known as \textit{Parseval's Theorem}). On the other hand, in computer science, we are usually interested in analyzing the special class of \textit{boolean functions} $f: \{0,1\}^n \to \{0,1\}$. These functions exhibit rich structure that Parseval's Theorem isn't equipped to capture---to understand them, we usually need to look beyond the second moment.

Hypercontractivity, introduced in 1970 by Bonami \cite{bonami1970etude} (and later independently by Beckner \cite{beckner1975inequalities} and Gross \cite{gross1975logarithmic}), is exactly the tool for the job. In its simplest form, hypercontractivity boils down to the statement that the fourth moment of low levels of the Fourier decomposition should behave nicely. Namely that the $i$th level of a boolean function $f$, denoted $f_i$, should satisfy: 
\begin{equation}\label{eq:intro-bonami}
\norm{f_i}_4 \leq 2^{O(i)}\norm{f_i}_2.
\end{equation}
This deceptively simple observation, known in the above form as ``Bonami's Lemma'' \cite{bonami1970etude}, led to many landmark results including the KKL Theorem \cite{kahn1988influence}, noise-sensitivity of sparse functions \cite{kahn1988influence}, Friedgut's Junta Theorem \cite{friedgut1998boolean}, and Majority is Stablest \cite{mossel2005noise}. What's more, hypercontractivity (and its resulting applications) actually extend beyond the hypercube. After KKL's seminal work, many authors studied extensions and applications of hypercontractivity \cite{bourgain1992influence,talagrand1994russo,friedgut1996every,friedgut1998boolean}, but it wasn't until recently that tight analogs of \Cref{eq:intro-bonami} were developed for general product spaces \cite{keevash2019hypercontractivity} as well as for other structured domains such as the symmetric group \cite{filmus2020hypercontractivity} and Grassmannian \cite{subhash2018pseudorandom}. These extended domains differ from the hypercube in that they are only hypercontractive for special classes of pseudorandom functions, but are nevertheless responsible for an impressive set of applications including analogs of classical results, a variety of new sharp threshold theorems \cite{keevash2019hypercontractivity,lifshitz2019noise,keevash2021global}, and perhaps most famously the proof of the 2-2 Games Conjecture \cite{khot2017independent,dinur2018towards,dinur2018non,barak2018small,khot2018small,subhash2018pseudorandom}. Unfortunately, despite the stark similarities between these settings, no unified theory explaining the phenomenen exists. Further, all known techniques rely heavily on product structure or other strong forms of symmetry, which makes it difficult to approach the problem in more general settings. 

\paragraph{Fourier Analysis on HDX:} High dimensional expanders (HDX) are a class of robustly connected complexes that have seen an incredible amount of development and application throughout theoretical computer science in the past few years, most famously in coding theory \cite{dinur2021,EvraKZ20,jeronimo2021near,KaufmanO21amplified, KaufmanT21quantum,dikstein2020locally,jeronimo2020unique,dinur2019list} and approximate sampling \cite{anari2019log,alev2020improved,anari2020spectral,chen2020rapid,chen2021optimal,chen2021rapid,feng2021rapid,jain2021spectral,liu2021coupling,blanca2021mixing}, but also in agreement testing \cite{dinur2017high,dikstein2019agreement,kaufman2020local}, CSP-approximation \cite{alev2019approximating,bafna2020high}, and (implicitly) hardness of approximation \cite{khot2018small,subhash2018pseudorandom}. In this work, we study a central notion of high dimensional expansion called \textit{two-sided local-spectral expansion}, originally developed by Dinur and Kaufman \cite{dinur2017high} to build sparse agreement testers.
For simplicity, we’ll often refer to these objects just as \emph{local-spectral expanders}, but the reader should be aware we always refer to the two-sided variant, not the weaker one-sided variant commonly used in approximate sampling.

Interestingly, local-spectral expanders are actually known to admit a (nascent) theory of Fourier analysis \cite{dikstein2018boolean,kaufman2020high}. Initial works in this area have focused on the development and application of Fourier Decompositions and Parseval's Theorem, and while the existing theory does have a few interesting applications (e.g. an FKN theorem for HDX \cite{friedgut2002boolean,dikstein2018boolean}, efficient CSP-approximation \cite{alev2019approximating,bafna2020high}), it is subject to the same limitations as original second moment methods on the hypercube: they simply don't capture the richer structure of boolean functions. Let's consider a concrete and important example: the expansion of pseudorandom sets (an analog of ``sparse functions are noise-sensitive'' on the hypercube).\footnote{The connection lies in the fact that the noise-sensitivity result can equivalently be phrased as saying that small sets on the noisy hypercube are expanding.} Traditionally proved via hypercontractivity, a variant of this result on the Grassmannian recently led to the resolution of the 2-2 Games Conjecture \cite{subhash2018pseudorandom}. On the other hand, Bafna, Hopkins, Kaufman, and Lovett \cite{bafna2020high} showed that second moment methods cannot recover such a result. While they are able to recover some sort of characterization with these techniques, it necessarily decays as the dimension grows to infinity, becoming trivial in the regime useful for hardness of approximation---if we want to do better, it appears we need a theory of hypercontractivity.

This is easier said than done: local-spectral expanders look nothing like any object previously known to satisfy hypercontractivity. They can be sparse, asymmetric, and very far from products. Moreover, there are no known techniques for analyzing local-spectral expanders beyond the second moment.\footnote{We note that recent works in the sampling literature have considered entropic notions of high dimensional expansion, but the underlying assumptions are much stronger than local-spectral expansion.} Even DDFH and KO's Fourier decompositions are intrinsically tied to second moment methods, since they are defined by linear algebraic manipulation of the standard inner product. Surprisingly, it turns out that these barriers are not inherent, and can be removed with the introduction of just two new tools: a \textit{combinatorial} Fourier decomposition for HDX, and a new local-to-global method to replace reliance on product structure in the analysis of higher moments.

Our new decomposition is the natural analog of the standard Fourier decomposition on product spaces (often called the ``orthogonal'' or ``Efron-Stein'' decomposition). It is equivalent to old decompositions in an $\ell_2$-sense (and therefore shares all relevant $\ell_2$-based properties), but comes with a number of additional benefits: it has simple explicit and recursive forms, and it behaves nicely under restriction. This allows us to bring to bear much of the power of more traditional Fourier-analytic machinery, which often relies on these same properties. Historically, however, applying this machinery in a useful fashion has also required the underlying object to be a product, or to satisfy some other strong symmetry. Our second key observation is that while individual variables in a local-spectral expander may be highly correlated, they look independent \emph{on average}. More concretely, this means that in the analysis of \textit{expectations} (such as a higher moment), we are free to treat the underlying variables as independent even if they actually exhibit a very high level of correlation.

\paragraph{Hypercontractivity on HDX:} Leveraging these tools, we build a theory of hypercontractivity on HDX. Concretely, we prove that \Cref{eq:intro-bonami} holds on local-spectral expanders for an appropriate notion of pseudorandom functions---ones that are not concentrated in any local restrictions on the complex.\footnote{In the high dimensional expansion literature, these restrictions are known as \textit{links}.} Combined with BHKL's recent spectral analysis of \textit{higher order random walks} (which, for the moment, we'll think of as analogs of the noisy hypercube graphs or Hamming scheme), this leads to the resolution of a number of open questions in boolean function analysis. To start, we provide a tight characterization of (edge) expansion on higher order random walks, which, unlike previous methods \cite{bafna2020high}, \textit{does not decay with dimension}. This matches the version of the result on the Grassmannian which led to the resolution of the 2-2 Games Conjecture \cite{subhash2018pseudorandom}, and opens yet another avenue towards the use of HDX in hardness of approximation. We also introduce natural analogs\footnote{When applied to the embedding of the hypercube into a simplicial complex, these definitions return the standard notions.} of two classic Fourier-analytic notions: \textit{influence} and the \textit{noise operator}. Combining these with the above recovers tight variants of both the KKL Theorem and noise-sensitivity of sparse (or in this case pseudorandom) functions.

Beyond these concrete applications, hypercontractivity on HDX also has interesting implications in the broader context of discrete Fourier analysis and high dimensional expansion. For the former, our result gives the first general class of hypercontractive objects beyond products, and combined with bounded degree constructions \cite{lubotzky2005explicit,kaufman2018construction}, the first example of hypercontractivity over any sparse object at all.\footnote{Formally, the notion of ``sparseness’’ here requires a bit of explanation. Previous settings of hypercontractivity all have natural representations as dense complexes. For instance, the hypercube complex has $2n$ vertices and $2^n$ top-level faces, and even weakly hypercontractive sparsifications like the short code \cite{barak2012making} have quasi-polynomially many faces.} For the latter, our result suggests a new stratification among notions of local-spectral expansion. This requires some additional explanation. While local-spectral expanders were originally introduced only over simplicial complexes, they were quickly extended to more general settings such as the Grassmannian, or even to general ranked posets \cite{dikstein2018boolean}. While these classes of local-spectral expanders are essentially equivalent in an $\ell_2$ sense \cite{dikstein2018boolean,bafna2020high,kaufman2021local}, our analysis of the fourth moment crucially relies on simplicial structure. We conjecture that this is an inherent rather than technical barrier: only special classes of underlying objects (e.g. Grassmannian, simplicial complexes) satisfy hypercontractivity, and thereby lead to the strongest known form of spectral high dimensional expanders.

\section{Background}\label{sec:background}
Before stating our results more formally, we give a quick overview of the theory of local-spectral expanders and higher order random walks. Local-spectral expansion is a robust notion of connectivity on weighted hypergraphs introduced by Dinur and Kaufman \cite{dinur2017high} in the context of agreement testing. As is standard in the area, we will view $d$-uniform hypergraphs $H \subseteq {[n] \choose d}$ as (pure) \textbf{simplicial complexes}:
\[
X_H = X(0) \cup \ldots \cup X(d),
\]
where $X(d)=H$, $X(i) \subseteq {[n] \choose i}$ is given by downward closure, and $X(0)=\emptyset$. We note that this notation is off by one from much of the HDX literature which considers $X(i) \subseteq {[n] \choose i+1}$. This notation is standard in the topological literature (where an $i$-simplex indeed as $i+1$ points), but is less natural for our purely combinatorial work. 

Most recent work on high dimensional expansion is based on the \textit{local-to-global paradigm}, in which local properties of a complex are lifted to a desired global property (e.g. mixing or agreement testing). The main local structure of interest are called \textbf{links}. For every ``$i$-face'' $\tau \in X(i)$, the link of $\tau$ is the subcomplex obtained by restriction to faces including $\tau$:
\[
X_\tau = \{ \sigma~:~ \sigma \cap \tau = \emptyset, \ \sigma \cup \tau \in X\}.
\]
A simplicial complex is said to be a \textbf{$\bm{\gamma}$-local-spectral expander} if (the graph underlying) every link is a $\gamma$-spectral expander.\footnote{A graph is a $\gamma$-spectral expander if the second largest eigenvalue of its adjacency matrix is at most $\gamma$ in absolute value.}

\textbf{Higher order random walks} are an analog of the standard walk on expander graphs that moves between two vertices via an edge. Kaufman and Mass \cite{kaufman2016high} observed that this process can be applied at any level of a simplicial complex: one could move between edges via a triangle, or triangles via a pyramid. Formally, these walks are defined as a composition of \textbf{averaging operators}, objects that have become ubiquitous tools in the study of high dimensional expanders. Denote the space of functions $\{f: X(k) \to \mathbb{R}\}$ as $C_k$. For a function $f\in C_k$, the (level $k$) \textbf{Up} and \textbf{Down operators} lift and lower $f$ to level $k+1$ and $k-1$ respectively by averaging:
\begin{align*}
    U_kf(\tau) &= \underset{\sigma \subset \tau}{\mathbb{E}}[f(\sigma)],\\
    D_kf(\tau) &= \underset{\sigma \supset \tau}{\mathbb{E}}[f(\sigma)].
\end{align*}
It will often be useful to compose the down or up operators multiple times to move between levels $k$ and $i$, we denote this by $D^k_i = D_i \circ \ldots \circ D_k$ and $U^k_i = U_k \circ \ldots \circ U_i$. Informally, HD-walks are simply affine combinations of composed averaging operators. For instance, the basic composition $N_k^i = U^{k+i}_kD^{k+i}_k$, called a \textbf{canonical walk}, is the random process which moves between two $k$-faces via a shared $(k+i)$-face.
\section{Results}
We now move to an informal description of our results. We view our work as having three main contributions. First, we introduce and develop a new theory of Fourier analysis on high dimensional expanders. This includes a new explicit Fourier decomposition, as well as a number of natural generalizations of Fourier-analytic ideas such as influence and the noise operator to simplicial complexes. Second, we prove that our Fourier-analytic decomposition satisfies a hypercontractive inequality for the special subclass of \emph{pseudorandom functions}, and use this fact to characterize the small set expansion of HD-walks and give a version of Bourgain's Theorem (an analog of KKL on product spaces) on HDX. Finally, en route to our hypercontractivity theorem, we introduce a new method of localization on high dimensional expanders of independent interest that enables local-to-global analysis of higher order moments.

\subsection{The Bottom-Up Decomposition}
We start with a discussion of our new explicit Fourier-analytic decomposition. All previously known Fourier bases on local-spectral expanders \cite{kaufman2020high,dikstein2018boolean} are linear algebraic in nature, and have no known closed form. While these decompositions certainly have their place and are sufficient for a number of interesting applications \cite{dikstein2018boolean,alev2019approximating,bafna2020high}, they often fall short when finer-grained calculation is required. To alleviate this issue, we introduce a new combinatorial decomposition on simplicial complexes that is an analog of the classic orthogonal (sometimes called Efron-Stein) decomposition on product spaces. For simplicity, we'll start by introducing the decomposition in a simple recursive form, and give the explicit form later in this section.
\begin{definition}[Bottom-Up (Recursive Form)]\label{intro:def-recursive}
Let $X$ be a $d$-dimensional pure simplicial complex and $f \in C_k$ any function. Recursively define the $i$th level function(s) to be:
\[
\gup{i} = D^k_i f - \sum\limits_{j=0}^{i-1} {i \choose j}U^i_j\gup{j}, \quad \fup{i} = {k \choose i}U^k_i\gup{i}
\]
We call $f = \sum\limits_{i=0}^k \fup{i}$ the Bottom-Up Decomposition.
\end{definition}
Strictly speaking, while the consideration of this basis is new over general simplicial complexes, it was first studied for the special case of the complete complex by \cite{khot2018small}. There, the authors took advantage of the complex's near-product structure to show that the decomposition gives an (approximate) Fourier basis close to the eigendecomposition of $f$ with respect to the well-studied Johnson graphs. We prove that the assumption of near-product structure is actually unnecessarily strong---it is enough for the underlying complex to be sufficiently expanding.

\begin{theorem}[Bottom-Up Properties (Informal \Cref{lemma:approx-orthog}+\Cref{thm:Bottom-vs-Top})]\label{thm:intro-BU-properties}
Let $X$ be a two-sided $\gamma$-local-spectral expander, and $M$ an HD-walk. Then for any $f \in C_k$, and $0 \leq i < j \leq k$:
\begin{enumerate}
    \item $\langle \fup{i},\fup{j} \rangle \approx 0$
    \item $\norm{f}_2^2 \approx  \sum\limits_{i=0}^k \norm{\fup{i}}_2^2$
    \item $\exists \lambda_i \text{ s.t. } M\fup{i} \approx \lambda_i \fup{i}$
\end{enumerate}
\end{theorem}
\Cref{thm:intro-BU-properties} is also very similar to an analogous result for the HD-Level-Set Decomposition in \cite[Theorem 1.3]{dikstein2018boolean}. We will cover their definition in greater detail in \Cref{sec:Bottom-Up}. For the moment, it suffices to note that their decomposition also breaks $f$ into $k+1$ Fourier levels, which we similarly denote by $f = \sum\limits_{i=0}^k \fdown{i}$. It turns out that the similarities between the HD-Level-Set and Bottom-Up Decompositions are no accident---the two decompositions are actually nearly equivalent.
\begin{theorem}[Bottom-Up Approximates HD-Level-Set (\Cref{thm:Bottom-vs-Top})]\label{thm:intro-Bottom-vs-Top}
Let $X$ be a two-sided $\gamma$-local-spectral expander and $f \in C_k$. Then the Bottom-Up and HD-Level-Set Decomposition are close in $\ell_2$-norm:
\[
\norm{\fup{i} - \fdown{i}}_2^2 \leq 2^{O(k)}\gamma\norm{f}^2_2.
\]
Similarly,
\[
\left | \langle \fup{i},\fup{i} \rangle - \langle \fdown{i},\fdown{i} \rangle \right | \leq 2^{O(k)}\gamma\norm{f}^2_2.
\]
\end{theorem}
The main advantage of the Bottom-Up Decomposition then lies in its simple recursive structure and explicit form, which we now describe.
\begin{proposition}[Bottom-Up Explicit Form (\Cref{thm:explicit})]
Let $X$ be a $d$-dimensional pure simplicial complex, $f \in C_k$ any function, and $\gup{i}$ be as given in \Cref{intro:def-recursive}. Then:
\[
\gup{i} = \sum_{j=0}^i (-1)^{i-j} \binom{i}{j} U_j^i D_j^k f,
\]
or equivalently $\forall \tau \in X(i)$:
\[
\gup{i}(\tau) = \sum\limits_{\sigma \subseteq \tau}(-1)^{|\tau \setminus \sigma|}\underset{X_\sigma}{\mathbb{E}}[f]
\]
\end{proposition}
In \Cref{sec:Bottom-Up}, we will see how these properties are useful for analyzing finer-grained structures like restriction that are often key to classical Fourier-analytic arguments. It is unknown how to analyze such properties for prior linear algebraic decompositions, and determining whether the latter share similar structure at this level remains an interesting open problem.

\subsection{Hypercontractivity}
Now that we have introduced our relevant Fourier-analytic decomposition, we turn our attention to the study of hypercontractivity. Hypercontractivity is one of the most powerful tools in boolean function analysis and is crucial to proving many of area's key results (e.g.\ KKL \cite{kahn1988influence}, FKN \cite{friedgut2002boolean}, Majority is Stablest \cite{mossel2005noise}, sharp threshold theorems \cite{friedgut1999sharp}, etc.). Informally, hypercontractivity can be thought of as a niceness condition on ``low-degree'' functions. We'll start by considering a simple variant often called the Bonami or Bonami-Beckner lemma, \cite{bonami1970etude} which states that a ``degree-$i$'' function $p$ should satisfy:
\[
\norm{p}_4 \leq 2^{O(i)}\norm{p}_2.
\]
Classically, we might think of $p$ as being a degree-$i$ polynomial, corresponding to the $i$th Fourier level of a boolean function. The corresponding statement in our context is therefore that the $i$th level of the Bottom-Up Decomposition should satisfy an analogous inequality:
\begin{equation}\label{eq:bonami}
    \norm{\fup{i}}_4 \leq 2^{O(i)}\norm{\fup{i}}_2.
\end{equation}
Unfortunately, it is well known that \Cref{eq:bonami} cannot hold in our setting, even over the complete complex. However, it is possible that the inequality could hold for \textit{natural subclasses} of functions. Indeed, such a phenomenon is known to occur on general product distributions \cite{keevash2019hypercontractivity}, where \textbf{pseudorandom} functions satisfy a form of \Cref{eq:bonami}.
\begin{definition}[Pseudorandomness]\label{def:intro-pr}
Let $X$ be a simplicial complex and $f \in C_k$. We say $f$ is $(\varepsilon,i)$-pseudorandom if it is sparse in every $i$-link in the following two senses:
\begin{enumerate}
    \item For all $\tau \in X(i)$: 
    \[
    \left| \underset{X_\tau}{\mathbb{E}}[f] \right| \leq \varepsilon \norm{f}_\infty
    \]
    \item For all $\tau \in X(i)$: 
    \[
    \langle f|_\tau, f|_\tau \rangle \leq \varepsilon \norm{f}_\infty^2
    \]
\end{enumerate}
\end{definition}

In applications, we will often only care about non-negative functions, in which case the second condition can be removed completely (as it is implied by the first). We note that functions satisfying \Cref{def:intro-pr} are also sometimes called \textit{global} since they are not concentrated in any local structure \cite{keevash2019hypercontractivity,lifshitz2019noise}. We call them pseudorandom in keeping with prior literature on the Johnson and Grassmann graphs \cite{khot2018small,subhash2018pseudorandom}, and because they cannot be distinguished from an ($\varepsilon$-sparse) random function by examining density inside links. Finally, note that \Cref{def:intro-pr} requires $f$ to be sparse. We conjecture that our results should hold in the dense regime as well, and discuss this further in \Cref{sec:discussion}.

Hypercontractivity for restricted subclasses is still a very powerful tool. KLLM's result, for instance, led to the resolution of Majority is Stablest in the $p$-biased setting \cite{lifshitz2019noise}, and the resolution of several conjectures in extremal combinatorics as well \cite{keevash2021global}. Unfortunately, their technique cannot be used to analyze simplicial complexes. Unlike the complete complex (which is essentially a product space), bounded degree HDX are necessarily quite far from products. Nevertheless, we prove that a variant of \Cref{eq:bonami} continues to hold.
\begin{theorem}[Hypercontractivity on HDX (Informal \Cref{thm:hypercontraction})]\label{thm:intro-hypercontraction}
Let $X$ be a sufficiently strong two-sided $\gamma$-local-spectral expander and $f \in C_k$ an $(\eps,i)$-pseudorandom function. Then the following hypercontractive inequality holds:
\[
\E[\fup{i}^4] \leq 2^{O(i)}\eps\E[\fup{i}^2]\norm{f}_\infty^2 + c_k\gamma^{1/2}\varepsilon\norm{f}_2^2\norm{f}_\infty^2,
\]
where $c_k \leq \min\{2^{O(k)},k^{O(i)}\}$.
\end{theorem}
Our overall framework for proving \Cref{thm:intro-hypercontraction} roughly follows KMMS' strategy for the complete complex. However, even with analogous results for the Bottom-Up Decomposition in hand, most of their techniques fail dramatically in our setting due to local-spectral expanders' distinct lack of product structure. In fact, \Cref{thm:intro-hypercontraction} gives the first general class of hypercontractive objects beyond product spaces, and combined with known bounded degree constructions of local-spectral expanders \cite{lubotzky2005explicit,kaufman2018construction}, the first example over any \textit{sparse} domain at all. In \Cref{sec:intro-local}, we'll discuss how we tackle these traditionally hard-to-handle structures with the introduction of a new notion of average-case independence that relates closely to local-spectral expansion. Our method actually allows for analysis well beyond the $4$th moment, and can also be used to extend \Cref{thm:intro-hypercontraction} to $2$-to-$2q$ hypercontractivity (where the $4$-norm is replaced by a higher $2q$-norm). We focus on the $2$-to-$4$ case in this work for simplicity.

Before moving on to applications of \Cref{thm:intro-hypercontraction}, it is worth discussing another typical form of hypercontractivity and how it translates to the setting of simplicial complexes. Hypercontractivity is frequently expressed in terms of an object called the \textbf{noise operator}. On the hypercube, the noise operator $T_\rho$ acts as an averaging process on boolean strings which replaces each coordinate with a random bit with probability $1-\rho$. In this context, hypercontractivity states that $T_\rho$ should act as a \textit{smoothing operator} in the following sense:
\begin{equation}\label{eq:noise-hyp}
\norm{T_\rho f}_4 \leq \norm{f}_2
\end{equation}
for some constant $\rho$. While it is not immediately obvious how to define $T_\rho$ over a simplicial complex, one can imagine a similar process where each vertex in a $k$-face is removed with probability $1-\rho$, and is then re-randomized over relevant $k$-faces. We formalize this procedure in terms of the averaging operators.
\begin{definition}[Noise Operator]
Let $X$ be a $d$-dimensional pure simplicial complex. The noise operator $T^k_\rho(X): C_k \to C_k$ at level $k \leq d$ of the complex is:
\[
T^k_\rho(X) = \sum\limits_{i=0}^k {k \choose i}(1-\rho)^i \rho^{k-i}U_{k-i}^kD^k_{k-i}.
\]
We write just $T_\rho$ when clear from context.
\end{definition}
When applied to the hypercube complex,\footnote{The hypercube complex has vertex set $[n] \times \{0,1\}$, where the first entry stands for a coordinate and the second entry a value. The top level $X(n)$ consists of all binary strings and is exactly the hypercube.} this natural analog returns exactly the standard boolean noise operator $T_\rho$. Combining standard arguments with the spectral properties of the Bottom-Up Decomposition, we can also prove a variant of \Cref{eq:noise-hyp} for pseudorandom functions on HDX. To state this result, it will be useful to have a notion of degree: as on the hypercube, we say the degree of a function $f$ is the largest $i$ such that $\fup{i}$ is non-zero.
\begin{corollary}[Informal \Cref{thm:noise-hyper}]
Let $X$ be a sufficiently strong two-sided $\gamma$-local-spectral expander and $f \in C_k$ a degree $i$, $(\delta,i)$-pseudorandom function for $\delta \leq \varepsilon\norm{\frac{f}{\norm{f}_\infty}}_2^2$. Then for some constant $\rho=\Theta(1)$:
\[
\norm{T_\rho f}_4 \leq \eps^{1/4}\norm{f}_2.
\]
\end{corollary}

\subsection{Applications}
A classical application of hypercontractivity is to give what is known as a ``level-$i$ inequality'' that bounds low-level weight of a boolean function. We can use \Cref{thm:intro-hypercontraction} to give an analog on HDX for pseudorandom functions.
\begin{theorem}[Level-$i$ inequality (Informal \Cref{thm:level-i})]\label{cor:level-i}
Let $X$ be a two-sided $\gamma$-local-spectral expander with $\gamma$ sufficiently small and $f \in C_k$ an $(\eps,i)$-pseudorandom boolean function of density $\alpha$. Then the weight on $\fup{i}$ is bounded by:
\[
\langle \fup{i},\fup{i} \rangle \leq 2^{O(i)}\eps^{1/3}\alpha.
\]
\end{theorem}
Level-$i$ inequalities have a plethora of applications in boolean Fourier analysis. We'll look at the analog of two classical applications: one to small-set expansion, and the other to the structure of functions with low influence. Starting with the former, let's recall the basic definition of edge-expansion.
\begin{definition}
Let $M$ be a walk on the $k$th level of a simplicial complex $X$. The (edge) expansion of a subset $S \subseteq X(k)$ is the average probability of leaving $S$ in a single step of the walk:
\[
\Phi(S) = \underset{v \sim S}{\mathbb{E}}[M(v,X(k) \setminus S)],
\]
where $M(v,X(k) \setminus S)$ is the probability the walk leaves $S$ starting from $v$.
\end{definition}
Informally, a walk is called a \emph{small-set expander} if all small subsets expand. Traditionally, the level-$i$ inequality on the discrete hypercube is used to show that the noisy hypercube graph is a small-set expander. The analogous result on simplicial complexes, however, isn't true: HD-walks (which generalize graphs like the noisy hypercube) have well-known examples of small non-expanding sets: \emph{links} \cite{khot2018small,bafna2020high}. Using \Cref{cor:level-i}, we can prove a converse to this result: \emph{any} non-expanding set must be concentrated in a link.

\begin{theorem}[Characterizing non-expansion on HD-walks (Informal \Cref{thm:expansion})]\label{intro:expansion}
For every $0 < \delta < 1$, there exists some $\varepsilon > 0$ and $r \in \mathbb{N}$ such that for all large enough $k$ the following holds. For any HD-walk\footnote{Formally, this statement only holds for HD-walks such as $N^{\Theta(k)}_{k}$ which exhibit sufficiently fast eigenvalue decay. We give a more general formulation in the main body that holds for all HD-walks (see \Cref{thm:expansion}).} on a sufficiently strong two-sided local-spectral expander $X$ and any subset $S \subseteq X(k)$, if $S$ has expansion at most $\Phi(S) \leq \delta$, then $S$ is concentrated in a low-level link:
\[
\exists i \leq r, s \in X(i): \frac{|X_s \cap S|}{|X_s|} \geq \varepsilon
\]
\end{theorem}

Expansion is also closely related to a well-studied Fourier-analytic quantity called \textbf{total influence}. On the boolean hypercube, the total influence of a function measures its total variability across each coordinate:
\[
I[f] = \sum\limits_{i=1}^n \Pr_{x \sim \{0,1\}^n}[f(x) = f(x \oplus e_i)]
\]
where $e_i$ is the $i$th standard basis vector. One of the most celebrated results in the analysis of boolean functions is the KKL Theorem \cite{kahn1988influence}, which states that any function with low total influence must have an influential coordinate. In domains beyond the hypercube (such as product spaces), total influence is usually instead written equivalently as:
\[
I[f] = \langle f,Lf \rangle
\]
where $L$ is the (un-normalized) \textbf{Laplacian operator} (see \Cref{sec:Fourier} for more details). While the KKL Theorem does not hold over product spaces (even say the $p$-biased hypercube), a useful analog known as ``Bourgain's Sharp Threshold Theorem'' \cite[Appendix]{friedgut1999sharp} does.  Bourgain's Theorem states that if a boolean function has small total influence, there must exist a link (on the hypercube a subcube) in which the function is much denser than expected.

We prove an analogous result for HDX. The Laplacian formulation of total influence has a natural generalization on simplicial complexes: 
\[
I_X[f] = \langle f, k(I-U_{k-1}D_k)f \rangle
\]
that returns the standard definition over the hypercube complex (see \Cref{sec:Fourier}). Using \Cref{cor:level-i}, we prove that any function with low total influence must be concentrated in a link.
\begin{theorem}[Bourgain's Theorem for HDX (Informal \Cref{thm:hdx-sharp-threshold})]\label{intro:hdx-sharp-threshold}
Let $X$ be a sufficiently strong two-sided $\gamma$-local-spectral expander, and $f \in C_k$ a boolean function. Then for any $0 \leq K \leq k$, if $I[f] \leq K\var(f)$, there exists $i \leq K$ and an $i$-face $\tau$ such that the link of $\tau$ is dense:
\[
 \underset{X_\tau}{\E}[f] \geq 2^{-O(K)}.
\]
\end{theorem}
Note that \Cref{thm:hdx-sharp-threshold} is actually a bit weaker than Bourgain's Theorem in the sense that it only promises a link that is much denser than average when the function $f$ is sparse. We conjecture that this result should hold in the dense regime as well (see \Cref{sec:discussion} for details). On the other hand, unlike Bourgain's Theorem (which has a density increase of $2^{-O(K^2)}$ rather than $2^{-O(K)}$ for general functions), our result is tight.\footnote{A similar tight version of Bourgain's Theorem for sparse functions on the $p$-biased cube was proved by \cite{keevash2019hypercontractivity}.}
\begin{proposition}[Bourgain's Theorem Lower Bound (Informal \Cref{thm:lower})]
Let $c \geq 1$ be any constant and $K>1$ an integer. For all $K \ll k \ll n$, there exists a Boolean function $f \in C_k$ on the $k$-dimensional complete complex on $n$ vertices satisfying:
\begin{enumerate}
    \item The influence of $f$ is small:
    \[
    I[f] \leq K\text{Var}(f).
    \]
    \item For every $i \leq cK$, all $i$-links are sparse:
    \[
\forall i \leq cK, \tau \in X(i): \underset{X_\tau}{\mathbb{E}}[f] \leq 2^{-\Omega(K)}.
\]
\end{enumerate}
\end{proposition}

\subsection{Localization (Average Independence)}\label{sec:intro-local}
Our hypercontractive inequality is derived from a new method of localization on high dimensional expanders of independent interest. Localization itself is of course not new---indeed such techniques have recently become synonymous with HDX. However, most prior work in the literature focuses on the localization of \emph{second moments}, whereas hypercontractivity requires the analysis of \emph{higher moments}. Traditionally, analysis beyond the second moment is difficult on HDX due to an inherent lack of product structure. We show that this can often be circumvented by a new method of decorrelating variables. \begin{theorem}\label{intro:localization}
Let $X$ be a $d$-dimensional two-sided $\gamma$-local-spectral expander and $f \in C_k$. Then for any $j \leq d-k$ and $\tau \in X(j)$, the global and localized expectation of $f$ over $X_\tau$ differ by an operator with small spectral norm:
\[
\underset{X_{\tau}(k)}{\E}[f] - \underset{X(k)}{\E}[f] = \Gamma f(\tau)
\]
where $\Gamma:C_k \to C_j$ satisfies $||\Gamma||\leq O_{k,j}(\gamma)$.
\end{theorem}
We emphasize that the first expectation in this definition is given by \emph{localizing} rather than \emph{restricting} $f$. In other words we are averaging over $k$-faces in the link $X_\tau$ (which are $(k+j)$-faces in the original complex) rather than over $k$-faces in the original complex $X$ that contain $\tau$. This latter notion of restriction is also very important in analysis of HDX. We discuss further in \Cref{sec:localization}.

\Cref{intro:localization} should really be thought of as saying that, on average, $f$ can be decorrelated from ``irrelevant'' $j$-faces that don't appear in the input. This is particularly useful when analyzing objects like HDX with high correlation. To understand the technique a bit more concretely, let's look at a basic example application. 

Let $X$ be a $\gamma$-local-spectral expander. We will often be interested in analyzing certain expected products on $X$. For instructive purposes, let's take a look at an example of such a product with just two instances of some $g \in C_2$:
\begin{equation}\label{eq:localization}
\underset{a \sim X(1)}{\mathbb{E}}\underset{b \sim X_a(1)}{\mathbb{E}}\underset{c \sim X_{ab}(1)}{\mathbb{E}}\big[g(a,b)g(a,c)\big] = \underset{a \sim X(1)}{\mathbb{E}}\underset{b \sim X_a(1)}{\mathbb{E}}\big[g(a,b)\underset{c \sim X_{ab}(1)}{\mathbb{E}}[g(a,c)]\big].
\end{equation}
Notice that if we were working over a product space, the distribution of $c \sim X_{ab}(1)$ would be the same as the distribution of $c \sim X_a(1)$. This allows us to significantly simplify the above:
\begin{align*}
    \underset{a \sim X(1)}{\mathbb{E}}\underset{b \sim X_a(1)}{\mathbb{E}}\big[g(a,b)\underset{c \sim X_{ab}(1)}{\mathbb{E}}[g(a,c)]\big] &= \underset{a \sim X(1)}{\mathbb{E}}\left[    \underset{b \sim X_a(1)}{\mathbb{E}}[g(a,b)]    \underset{c \sim X_{a}(1)}{\mathbb{E}}[g(a,c)] \right]\\
    &= \underset{a \sim X(1)}{\mathbb{E}}\left[    \underset{b \sim X_a(1)}{\mathbb{E}}[g(a,b)]^2 \right].
\end{align*}
On the other hand in an HDX (especially one of bounded degree), this could be far from true since $b$ and $c$ can be highly correlated. \Cref{intro:localization} provides a simple technique for circumventing this issue. Let $g|_{a}$ be the restriction of $g$ to $a$, that is $g|_{a}(b) = g(a,b)$. \Cref{intro:localization} promises that
\[
\underset{c \sim X_{ab}(1)}{\mathbb{E}}[g(a,c)] =\underset{c \sim X_{a}(1)}{\mathbb{E}}[g(a,c)] + \Gamma g|_{a}(b),
\]
where $\norm{\Gamma} \leq O(\gamma)$. This allows us to recover the same form as above up to $O(\gamma)$ error:
\begin{align*}
  \underset{a \sim X(1)}{\mathbb{E}}\left[    \underset{b \sim X_a(1)}{\mathbb{E}}[g(a,b)]    \underset{c \sim X_{ab}(1)}{\mathbb{E}}[g(a,c)] \right] &= \underset{a \sim X(1)}{\mathbb{E}}\left[    \underset{b \sim X_a(1)}{\mathbb{E}}[g(a,b)]^2 \right] + \underset{a \sim X(1)}{\mathbb{E}}\left[    \underset{b \sim X_a(1)}{\mathbb{E}}[g(a,b) \cdot \Gamma g(b)] \right]\\
  &\leq \underset{a \sim X(1)}{\mathbb{E}}\left[    \underset{b \sim X_a(1)}{\mathbb{E}}[g(a,b)]^2 \right] + O_g(\gamma),
\end{align*}
where we have ignored some terms in $g$ for simplicity and the last step follows from an application of Cauchy-Schwarz and the spectral norm (see \Cref{sec:hyper} for details).

We emphasize that while \Cref{eq:localization} in particular could also have been analyzed through a more direct application of the swap walk, \textit{these standard techniques fail miserably when additional copies of $g$ are added}. Since there will generally be $j$ copies of $g$ in analysis of the $j$th moment, this means that traditional techniques cannot go beyond the second moment. On the other hand, our technique is applied individually to each copy of $g$, so it is essentially irrelevant how many times it appears in the product.

\section{Discussion}\label{sec:discussion}
Before getting into the details and formalization of the above, we take a moment to give a more careful treatment of some interesting open problems and related work.
\subsection{Open Problems}
Hypercontractivity, both on the cube and on extended domains, has led to an astounding number of applications since its introduction some 50 years ago. We recover just a small sample of these classical applications in our work, and believe the theory will give rise to further results in the analysis of boolean functions. However, rather than surveying a list of classical results one might wish to extend (we refer the reader to O'Donnell's excellent book \cite{o2014analysis} for this), we'll instead focus on three open problems we feel are most directly raised by our work.

Perhaps the most obvious direction left open is to extend hypercontractivity to the \textit{dense regime}. While our definition of pseudorandomness implicitly assumes the underlying function is sparse, we conjecture that all of our results should hold under a weaker notion of pseudorandomness that drops this assumption.
\begin{definition}[Pseudorandomness (Dense Regime)]
Let $X$ be a simplicial complex and $f \in C_k$ a boolean function. We say $f$ is $(\varepsilon,i)$-pseudorandom if its local and global average are close on every $i$-link:
\[
\forall \tau \in X(i): \big | \underset{X_\tau}{\mathbb{E}}[f] - \mathbb{E}[f] \big | \leq \varepsilon.
\]
\end{definition}
While the stronger notion we use in this work is certainly sufficient for some applications (e.g.\ characterizing expansion, noise-sensitivity) and is line with previous work \cite{khot2018small,subhash2018pseudorandom,keevash2019hypercontractivity}, it does seem to fall short in other areas. A good example of this is our variant of Bourgain's Theorem. While our version only promises the existence of a dense link, the original result on product spaces actually promises a link with \textit{higher than average density} (albeit by a factor of $2^{-O(K^2)}$ instead of $2^{-O(K)}$), which could be recovered by proving hypercontractivity for the above definition. More generally, proving hypercontractivity for this dense variant opens the door to a broader spectrum of applications than the sparse regime alone can handle.

The second problem we'd like to discuss is more focused on the theory of high dimensional expanders itself. As mentioned in the introduction, local-spectral expansion can be extended well beyond simplicial complexes to many natural poset structures including the Grassmann poset \cite{dikstein2018boolean,kaufman2021local}, where hypercontractivity was crucial to resolving the 2-2 Games Conjecture \cite{subhash2018pseudorandom}. The spectral and $\ell_2$-structure of these expanding posets (eposets) is well understood \cite{dikstein2018boolean,alev2019approximating,bafna2020high}, and essentially has no dependence on the underlying poset structure.\footnote{Different poset parameters result in different eigenvalues, but the structure is otherwise the same.} In stark contrast, our results break down over general eposets at several key points. In fact, it seems likely that the Bottom-Up Decomposition is not even a Fourier basis (fails to satisfy \Cref{thm:intro-BU-properties}) over general eposets, since the proof relies heavily on simplicial structure (see \Cref{lemma:approx-kernel}). On the other hand, variants of hypercontractivity are known for some special eposets such as the Grassmann poset. The key difference in these cases is that the definition of pseudorandomness necessarily changes. This raises a natural question: do all eposets satisfy hypercontractivity for some notion of pseudorandomness, or are structures like the Grassman poset and simplicial complexes ``special''? We conjecture that the latter is the case, and that these objects represent a new, stronger class of spectral high dimensional expanders.

Our third proposed problem is not raised quite as directly by this work, but is hard to ignore in light of recent breakthroughs in approximate sampling via HDX \cite{anari2019log,alev2020improved,anari2020spectral,chen2020rapid,chen2021optimal,chen2021rapid,feng2021rapid,jain2021spectral,liu2021coupling,blanca2021mixing}. Hypercontractivity is classically connected to the \textit{Log-Sobolev inequality}, which gives strong control over the mixing time of its associated random walk. Applied to the hypercube, for instance, this connection improves the standard spectral mixing bound from $O(n^2)$ to the optimal $\Theta(n\log(n))$ \cite{diaconis1996logarithmic}. Recent analysis of entropic notions of high dimensional expansion and a \textit{modified Log-Sobolev inequality} have led to a slew of analogous improvements on important sampling problems \cite{chen2021optimal,blanca2021mixing,anari2021entropic}. These results, however, usually only apply to dense objects and need stronger assumptions. Given these connections, it is natural to ask whether our theory of hypercontractivity can improve mixing times for general local-spectral expanders in some analogous fashion.

\subsection{Related Work}

\paragraph{Hypercontractivity on Extended Domains:} Nearly 20 years after its introduction, Kahn, Kalai, and Linial \cite{kahn1988influence} revolutionized the study of boolean functions with hypercontractivity. Not long after, a significant interest grew in the development and application of hypercontractivity beyond the hypercube, with a particular focus on product distributions and especially the $p$-biased hypercube \cite{bourgain1992influence,talagrand1994russo,friedgut1996every,friedgut1998boolean}. These works offered a general theory of hypercontractivity for such domains, but their strength depended on the underlying distributions in the product space. While this was sufficient for many important applications including a number of sharp threshold theorems \cite{friedgut1996every,friedgut1999sharp}, the results become meaningless for unbalanced products such as the $p$-biased cube for $p \leq o_n(1)$. This issue was resolved only recently by Keevash, Lifshitz, Long, and Minzer \cite{keevash2019hypercontractivity}, who showed a hypercontractivity theorem for pseudorandom functions on product spaces that is independent of the underlying distributions. This result offered the missing piece for a number of classical applications, including a tight variant of the KKL Theorem (for monotone functions) \cite{keevash2019hypercontractivity}, Majority is Stablest \cite{lifshitz2019noise}, and had a number of other applications in extremal combinatorics \cite{keevash2021global}.

Another line of work has examined hypercontractivity on what are often called ``exotic'' domains: specific objects beyond products such as the slice \cite{khot2018small}, multislice \cite{filmus2018log}, Grassmannian \cite{subhash2018pseudorandom} (or similarly the degree-two short code \cite{barak2018small}), and symmetric group \cite{filmus2020hypercontractivity}. Like KLLM's improved result for product distributions, most of these examples are only hypercontractive for pseudorandom functions (with the multislice being the only exception).\footnote{We note that higher degrees of the short code are also hypercontractive, but only on low Fourier levels for general functions \cite{barak2012making}.} The main application of this line of work has been to agreement testing and hardness of approximation. In particular, hypercontractivity for the Grassmannian was used to prove the soundness of an agreement tester in the ``1\% regime'' needed for the proof of the 2-2 Games Conjecture \cite{khot2017independent,dinur2018towards,dinur2018non,barak2018small,khot2018small,subhash2018pseudorandom}. It is worth noting that agreement testing theorems are also known for local-spectral expanders \cite{dinur2017high,dikstein2019agreement,kaufman2020local} (indeed the objects were originally introduced in this context). These results, however, lie in the ``99\% regime,'' so it is interesting to ask whether our theory of hypercontractivity can be used to build a bounded degree agreement tester in the more difficult 1\% regime.

Finally, we should note that our overarching proof structure for hypercontractivity builds on KMMS' work on the slice (i.e.\ the complete complex). Their techniques, however, rely heavily on the fact that the slice is close in $\ell_1$-distance to a product. This is far from true on local-spectral expanders, especially those of bounded degree which may essentially be as far as possible from products. As previously discussed, this lack of structure is a challenging barrier broken for the first time in this work.

\paragraph{Fourier Analysis on HDX:} Fourier analysis on HDX was originally studied by Diksein, Dinur, Filmus, and Harsha \cite{dikstein2018boolean}, who introduced the HD-Level-Set Decomposition, analyzed its spectral properties, and used it to prove an FKN Theorem for HDX. A similar decomposition was also proposed around the same time by Kaufman and Oppenheim \cite{kaufman2020high}, though their work was more focused on understanding the spectral structure of higher order random walks than on developing a theory of Fourier analysis. In the years since, the HD-Level-Set Decomposition has seen some further development \cite{alev2019approximating,kaufman2020chernoff,bafna2020high}, and the nascent theory has helped build efficient approximation algorithms for certain $k$-CSPs \cite{alev2019approximating} and unique games \cite{bafna2020high}, but the restriction to second moment methods seems to have limited its use otherwise. Towards breaking this same barrier, Gur, Lifshitz, and Liu \cite{GurCommunication} have also (independently) developed a similar theory of hypercontractivity on local-spectral expanders. While their work certainly shares some connections to ours, its main proof techniques differ substantially and we believe the two works are of independent interest.

\subsection{Roadmap}
Having concluded introductory discussion of our work, we lay out a brief roadmap for the rest of the paper. In \Cref{sec:prelims} we give preliminaries and formally define local-spectral expansion and higher order random walks. In \Cref{sec:localization} we discuss our new local-to-global method for higher moments that allows us to move beyond product distributions. In \Cref{sec:Bottom-Up} we discuss our new explicit Fourier Decomposition, its basic properties, and behavior under restriction. In \Cref{sec:hyper} we prove hypercontractivity for pseudorandom functions (\Cref{thm:intro-hypercontraction}). In \Cref{sec:expansion} we apply this result to characterize edge expansion in HD-walks (\Cref{intro:expansion}). Finally in \Cref{sec:Fourier} we introduce analogs of classic Fourier analytic notions such as influence and the noise operator and use them to prove both a KKL Theorem (\Cref{intro:hdx-sharp-threshold}) and noise-sensitivity of pseudorandom functions.

\section{Preliminaries}\label{sec:prelims}
Before moving into proofs and further discussion of our main results, we take a moment to cover the theory of local-spectral expanders and higher order random walks in more detail.
\subsection{Simplicial Complexes}
Our main objects of interest in this work are a family of expanding hypergraphs known as \textbf{local-spectral expanders}. In this context, it will be useful to think of $d$-uniform hypergraphs as objects called \textbf{pure simplicial complexes}.
\begin{definition}[Weighted, Pure Simplicial Complex]
A $d$-dimensional, pure simplicial complex $X=X(0) \cup \ldots \cup X(d)$ on $n$ vertices is the downward closure of a hypergraph $X(d)={[n] \choose d}$ where
\[
X(i) = \left \{ s \in {[n] \choose i}~\bigg\rvert~\exists t \in X(d), s \subseteq t \right \}.
\]
We call the elements of $X(i)$ $i$-faces. A weighted pure simplicial complex $(X,\Pi)$ is a simplicial complex $X$ endowed with a distribution $\Pi$ over $X(d)$. This induces a distribution over each $X(i)$ by downward closure:
\begin{equation}\label{eq:distr-down-clos}
\pi_i(x) = \frac{1}{i+1}\sum\limits_{y \in X(i+1):y \supset x}  \pi_{i+1}(y),
\end{equation}
where $\pi_d = \Pi$.
\end{definition}
Weighted pure simplicial complexes are equivalent to weighted hypergraphs, and we will adopt the former viewpoint throughout the rest of this work. We note that our definition of dimension is off by one from some of the literature which adopts the convention that an $i$-face has $i+1$ vertices. While this is natural from a topological viewpoint, it makes less sense in our combinatorial context.

Weighted simplicial complexes also come equipped with a natural set of inner products. Recall that $C_i = C_i(X)$ denotes the space of functions $f: X(i) \to \mathbb{R}$. The distribution $\Pi=(\pi_d,\ldots,\pi_0)$ induces a natural inner product on each level:
\[
\forall f,g \in C_i: \quad \langle f,g \rangle_{X(i)} = \underset{\tau \sim \pi_i}{\mathbb{E}}[f(\tau)g(\tau)].
\]
When clear from context, we drop $X(i)$ from the notation. Just like on the hypercube, these associated products are a core component of function analysis and the development of Fourier analysis on HDX.

\subsection{Local Spectral Expansion}
In this work we focus on a recent spectral notion of high dimensional expansion called \textit{two-sided local-spectral expansion} introduced by Dinur and Kaufman \cite{dinur2017high}. The definition hinges crucially on a form of local structure in simplicial complexes called \textbf{links}.
\begin{definition}[Link]
Let $(X,\Pi)$ be a $d$-dimensional weighted, pure simplicial complex. The \textbf{link} of an $i$-face $s \in X(i)$ is a $(d-i)$-dimensional pure simplicial complex given by the restriction of $X$ to faces containing $s$, that is:
\[
X_s = \{t \setminus s \in X~|~ t \supseteq s\}.
\]
We call $X_s$ an \textbf{$i$-link}. Throughout the rest of the paper, $X_s$ will always refer to its weighted version $(X_s,\Pi_s)$ where $\Pi_s$ is induced by the original distribution $\Pi$ by normalizing over top level faces of $X_s$.
\end{definition}
When analyzing a particular level $k$ of the complex, we will often abuse notation and write $X_s$ to mean the set of $k$-faces in $X$ which contain $s$ when clear from context.

Much of the high dimensional expansion literature centers around what is called the \textbf{local-to-global paradigm}, where properties on links are lifted to global properties on a complex. Local-spectral expansion can be seen as a definitional formalization of this notion: a complex is said to be expanding if all its local parts are expanding.
\begin{definition}[Local-spectral expansion \cite{dinur2017high}]
A weighted, pure simplicial complex $(X,\Pi)$ is a two-sided $\gamma$-local-spectral expander if for every $i \le d-2$ 
and every face $s \in X(i)$, the underlying graph\footnote{The underlying graph of a complex $X$ is $G=(V=X(1),E=X(2))$.} of $X_s$
is a two-sided $\gamma$-spectral expander.\footnote{A weighted graph is a two-sided $\gamma$-spectral expander if $\max\{|\lambda_2|,|\lambda_n|\} \leq \gamma$.}
\end{definition}

\subsection{Higher Order Random Walks}
Just like expander graphs are inextricably tied to their underlying random walks, local-spectral expanders are similarly connected to an analogous set of random processes known as higher order random walks (HD-walks). In \Cref{sec:background}, we discussed one example of these objects called the canonical walks that move between $k$-faces via a shared $(k+i)$-face, and saw that these could be defined by the averaging operators. We'll now extend these definitions to the more general setting of weighted simplicial complexes and, as well as define HD-walks in full generality. We'll start with the weighted averaging operators.
\begin{definition}[Averaging Operators]
Let $(X,\Pi)$ be a $d$-dimensional weighted, pure simplicial complex. For every $0 \leq k < d$, the \textbf{Up Operator} $U_k$ lifts functions from $C_k$ to $C_{k+1}$ by averaging:
\[
\forall \tau \in X(k+1): U_kf(\tau) = \frac{1}{k+1}\sum\limits_{\sigma \in X(k): \sigma \subset \tau}f(\sigma).
\]
Similarly, the \textbf{Down Operator} lowers functions from $C_{k+1}$ to $C_k$ by averaging:
\[
\forall \tau \in X(k): D_{k+1}f(\tau) = \frac{1}{\pi_{k+1}(X_\tau)}\sum\limits_{\sigma \in X_\tau} \pi_{k+1}(\sigma)f(\sigma),
\]
where $\pi_{k+1}(X_\tau) = \sum\limits_{\sigma \in X_\tau}\pi_{k+1}(\sigma)$, and the sum is over $k+1$ faces of $X$ containing $\tau$.
\end{definition}
It is worth noting that the averaging operators are \textit{adjoint} with respect to the associated inner products mentioned in the previous section, that is for any $f \in X(i)$ and $g \in X(i-1)$:
\[
\langle f,U_ig\rangle_{X(i)} = \langle D_if,g \rangle_{X(i-1)} = \underset{(\sigma,\tau) \sim (\pi_{i},\pi_{i-1})}{\mathbb{E}}[f(\sigma)g(\tau)].
\]
This means that basic combinations of the operators such as the canonical walks discussed in \Cref{sec:background} are \textit{self-adjoint} and therefore have a spectral decomposition.

Let's now formalize the notion of higher order random walks. We'll start with a basic version called \textbf{pure walks} that are simply a composition of the averaging operators.
\begin{definition}[Pure Walk \cite{alev2019approximating}]
Given a weighted, pure simplicial complex $(X,\Pi)$, a $k$-dimensional pure walk $Y: C_k \to C_k$ on $(X,\Pi)$ of height $h(Y)$ is a composition:
\[
Y = Z_{2h(Y)} \circ \cdots \circ Z_{1},
\]
where each $Z_i$ is a copy of $D$ or $U$.
\end{definition}
For the moment we won't force these walks to be self adjoint, but as we noted basic examples such as $N_k^i=D_{k}^{k+i}U_k^{k+i}$ do satisfy this constraint.

We define general higher order random walks to be any linear combinations of pure walks which is stochastic and self-adjoint.
\begin{definition}[HD-walk \cite{alev2019approximating}]\label{def:HD-walk}
Let $(X,\Pi)$ be a pure, weighted simplicial complex, and $\mathcal Y$ a family of pure walks $Y: C_k \to C_k$ on $(X,\Pi)$. We call a linear combination
\[
M = \sum\limits_{Y \in \mathcal Y} \alpha_Y Y
\]
a $k$-dimensional HD-walk on $(X,\Pi)$ as long as it is stochastic and self-adjoint. We call $w(M) \coloneqq \sum |\alpha_Y|$ the weight of $M$, and $h(M) = \max\{h(Y)\}$ its height.
\end{definition}

\subsection{Rectangular Swap Walks}
\Cref{def:HD-walk} only captures walks which stay on some fixed level of the complex. While these are certainly our main object of study, it turns out that in analysis it is often useful to consider \textit{rectangular walks} which move between levels of the complex. We will be particularly interested in a rectangular walk introduced independently by Alev, Jeronimo, and Tusiani \cite{alev2019approximating}, and Dikstein and Dinur \cite{dikstein2019agreement} called the \textit{swap walk}. Informally, the swap walk from $X(i)$ to $X(j)$ moves from an $i$-face $\tau$ to a $j$-face $\sigma$ through a shared $(i+j)$-face, but \textit{swaps out all original elements in $\tau$}. In other words, the intersection between $\tau$ and $\sigma$ must be empty, and the shared $(i+j)$-face is exactly $\tau \cup \sigma$. To formalize this, it is useful to first introduce the more basic rectangular walk moving between $X(i)$ and $X(j)$ with no such restrictions.
\begin{definition}[Rectangular Canonical Walks]
Let $(X,\Pi)$ be a $d$-dimensional, pure, weighted simplicial complex. For any $i+j \leq d$, the rectangular canonical walk $N_{i,j}$ is the natural operator moving between $X(i)$ and $X(j)$ through $X(i+j)$:
\[
N_{i,j} = D^{i+j}_{i}U^{i+j}_j.
\]
\end{definition}
Swap walks are then defined by forcing the down steps in a canonical walk to remove only vertices in the initial face.
\begin{definition}[Rectangular Swap Walks]
Let $(X,\Pi)$ be a $d$-dimensional, pure, weighted simplicial complex. For any $i+j \leq d$, the rectangular swap walk $S_{i,j}$ is the (normalized) restriction of $N_{i,j}$ to pairs $(\tau,\sigma) \in X(i) \times X(j)$ such that $|\tau \cap \sigma|=0$.
\end{definition}
Swap walks appear naturally in a number of areas, including agreement testing \cite{dikstein2019agreement}, coding theory \cite{jeronimo2021near}, and approximation algorithms \cite{alev2019approximating,bafna2020high} and were well studied even before their formal introduction on HDX. On the complete complex, for instance, rectangular swap walks are exactly the bipartite Kneser graphs. Swap walks are particularly useful in these contexts because unlike their canonical counterpart canonical walks, they are actually great expanders.
\begin{theorem}[Theorem 7.1 \cite{dikstein2019agreement}]\label{thm:dd}
Let $(X,\Pi)$ be a $d$-dimensional two-sided $\gamma$-local-spectral expander. Then for any $i+j \leq d$, the spectral expansion of $S_{i,j}$ is at most:
\[
\lambda(S_{i,j}) \leq ij\gamma,
\]
where $\lambda(S_{i,j})$ is the second largest singular value of $S_{i,j}$.
\end{theorem}
We note that this result was concurrently proved by AJT \cite{alev2019approximating}, albeit with a quantitatively worse bound.

\section{Localization Beyond the Second Moment}\label{sec:localization}
Localization is one of the (if not the) most important technique in the analysis of high dimensional expanders. Classic results, often grouped together under the name \textit{Garland's method}, show how global functions on simplicial complexes can be broken down into an average over local parts. There are two forms of Garland's method that will be relevant to our work. The first handles \textit{restrictions} of a function $f$ in $C_k$ to any $\tau \in X(i)$, s.t. $\restrict{f}{\tau} \in C_{k-i}(X_\tau)$ satisfies:
\[
\forall \sigma \in X_\tau(k-i): \restrict{f}{\tau}(\sigma) = f(\tau \cup \sigma).
\]
\begin{lemma}[Garland's method (restrictions) \cite{kaufman2020high}]\label{lemma:garland-restrict}
Let $(X,\Pi)$ be a weighted, pure simplicial complex, and $f \in C_k$. Then for any $i \leq k$, $\norm{f}^2$ is equal to its average second moment restricted to $i$-links:
\[
\langle f,f \rangle = \underset{\tau \in X(i)}{\E}[\langle \restrict{f}{\tau},\restrict{f}{\tau} \rangle].
\]
\end{lemma}
The second form of interest handles \textit{localizations} of $f$ to $\tau \in X(i)$, where $\localize{f}{\tau} \in C_k(X_\tau)$ lifts $f$ from $X(k)$ to $X_\tau(k)$:
\[
\forall \sigma \in X_\tau(k): \localize{f}{\tau}(\sigma) = f(\sigma).
\]
\begin{lemma}[Garland's method (localizations) \cite{oppenheim2018local}]\label{lemma:garland-local}
Let $(X,\Pi)$ be a $d$-dimensional, weighted, pure simplicial complex, and $f \in C_k$. Then for any $k+i \leq d$, $\norm{f}_2^2$ is equal to its average second moment localized to $i$-links:
\[
\langle f,f \rangle = \underset{\tau \in X(i)}{\E}[\langle \localize{f}{\tau},\localize{f}{\tau} \rangle].
\]
\end{lemma}
Garland's method will play an important role in the analysis of \Cref{thm:intro-hypercontraction}, but the results are generally only useful once a problem has been reduced to analyzing second moments. Since we are mainly interested in hypercontractivity and analyzing higher moments, Garland's method alone won't be sufficient. 

To this end, we introduce a new technique for analyzing higher moments on two-sided local-spectral expanders. At its core, the strategy relies on a deceptively simple observation: the difference between the global expectation of $f$ and its localized expectation over links is exactly given by an application of the swap walk minus its stationary operator.
\begin{lemma}\label{lemma:localization}
Let $(X,\Pi)$ be a $d$-dimensional pure, weighted simplicial complex and $f \in C_i$. Then for any $v \in X(j)$ such that $i+j \leq d$, we have:
\[
\underset{X_{v}}{\E}[\localize{f}{v}] - \E[f] = (S_{j,i} - U^j_0D_0^i)f(v).
\]
\end{lemma}
\begin{proof}
This is essentially immediate from expanding the left-hand side. We have:
\begin{align*}
    \underset{X_{v}}{\E}[\localize{f}{v}] - \E[f] &= \sum\limits_{w \in X(i)} \pi_{v,i}(w)f(w) - \sum\limits_{w \in X(i)}\pi_i(w)f(w)\\
    &= \sum\limits_{w \in X(i)} \left(\pi_{v,i}(w) - \pi_i(w)\right)f(w)\\
    &= (S_{j,i} - U^j_0D_0^i)f(v),
\end{align*}
where $\pi_{v,i}(w)=0$ for any $w \notin X_v(i)$.
\end{proof}
\Cref{lemma:localization} is particularly powerful on two-sided local-spectral expanders, since the spectral norm of $||S_{j,i} - U^j_0D_0^i||$ is small on every link by \Cref{thm:dd}.
\begin{corollary}\label{cor:localization}
Let $(X,\Pi)$ be a $d$-dimensional two-sided $\gamma$-local-spectral expander, $f \in C_i$, and $\tau \in X(\ell)$ for any $\ell <i$. Then for any $v \in X_{\tau}(j)$ such that $i+j-\ell \leq d$, the global and localized expectations of $f$ differ by:
\[
\underset{X_{\tau \cup v}}{\E}[\restrict{f}{\tau}] - \underset{X_{\tau}}{\E}[\restrict{f}{\tau}] = \Gamma \restrict{f}{\tau}(v)
\]
where $\Gamma: C_{i-\ell}(X_\tau) \to C_j(X_\tau)$ is an operator with spectral norm at most $||\Gamma||\leq (i-\ell)j\gamma$.
\end{corollary}

\begin{proof}
Applying \Cref{lemma:localization} to the restricted function $\restrict{f}{\tau}$ on $X_\tau$ (which is also a two-sided $\gamma$-local-spectral expander), we have that the left-hand side is exactly given by $(S_{j,i-\ell} - U_0^{j}D_0^{i-\ell})_{\tau}\restrict{f}{\tau}(v)$. Since $U_0^{j}D_0^{i-\ell}$ is the stationary operator of $S_{j,i-\ell}$, the spectral norm $\norm{S_{j,i-\ell} - U_0^{j}D_0^{i-\ell}}$ is exactly the second largest singular value of $S_{j,i-\ell}$. As discussed in \Cref{thm:dd}, Dikstein and Dinur \cite[Theorem 7.1]{dikstein2019agreement} proved that this quantity is at most $(i-\ell)j\gamma$ on any two-sided $\gamma$-local-spectral expander.
\end{proof}
\section{The Bottom-Up Decomposition}\label{sec:Bottom-Up}
In this section, we introduce the Bottom-Up Decomposition, an explicit combinatorial decomposition on simplicial complexes which approximates Dikstein, Dinur, Filmus, and Harsha's HD-Level-Set Decomposition \cite{dikstein2018boolean}. This is particularly useful since the latter decomposition essentially corresponds to the eigenspaces of HD-walks \cite{dikstein2018boolean,alev2019approximating,bafna2020high}.
\begin{definition}[Level Functions (Recursive Form)]\label{def:recursive}
Let $(X,\Pi)$ be a $d$-dimensional pure simplicial complex and $f \in C_k$ any function. The $i$th level function of the Bottom-Up Decomposition is given by
\[
\gup{i} = D^k_i f - \sum\limits_{j=0}^{i-1} {i \choose j}U^i_j\gup{j}
\]
\end{definition}
The Bottom-Up Decomposition is given by lifting the level functions via the up operator.
\begin{theorem}[Bottom-Up Decomposition (Explicit Form)]\label{thm:explicit}
Let $(X,\Pi)$ be a $d$-dimensional pure simplicial complex and $f \in C_k$ any function. Let $\fup{i}={k \choose i}U^k_i\gup{i}$ be the lift of the $i$th level function to $C_k$. Then the following statements hold:
\begin{enumerate}
    \item The lifted level functions give a decomposition of $f$:
    \[f = \sum\limits_{i=0}^k \fup{i}\]
    \item The lifted level functions have the following explicit form:
    \[
    \fup{i} = {k \choose i}\sum_{j=0}^i (-1)^{i-j} \binom{i}{j} U_k^i D_j^k f,
    \]
    or equivalently for all $\tau \in X(i)$:
    \[
    \gup{i}(\tau) = \sum\limits_{\sigma \subseteq \tau}(-1)^{|\tau \setminus \sigma|}\underset{X_\sigma}{\mathbb{E}}[f], \quad \fup{i}(\tau)=\sum_{\sigma \in X(i): \sigma \subset \tau} \gup{i}(\sigma)
    \]
\end{enumerate}
\end{theorem}
\begin{proof}
(1) can be proved directly by the explicit form given in (2):
\begin{align*}
\sum_{i=0}^k \fup{i} &= \sum_{i=0}^k \binom{k}{i} U_i^k \sum_{j=0}^i (-1)^{i-j} \binom{i}{j} U_j^i D_j^k f\\
&= \sum_{j=0}^k \left(\sum_{i=j}^k (-1)^{i-j} \binom{k}{i} \binom{i}{j}\right) U_j^k D_j^k f\\
&= U_k^k D_k^k f\\
&= f.
\end{align*}
where we've used the fact that:
\[
\sum_{i=j}^k (-1)^{i-j} \binom{k}{i} \binom{i}{j} = \delta_{jk}.
\]
It is left to prove (2). We proceed by induction. Note that the equality clearly holds for $i=0$, where both sides are simply the global expectation $\E[f]$. Now assume by induction that the equivalence holds up to $i-1$. We then have:
\begin{align*}
    \gup{i} &= D^k_i f - \sum\limits_{j=0}^{i-1}{i \choose j}U^i_j \gup{j}\\
    &= D^k_i f -\sum\limits_{j=0}^{i-1}{i \choose j}U^i_j\sum\limits_{\ell=0}^j (-1)^{j-\ell}{j \choose \ell}U^j_\ell D_\ell^k f\\
    &= D^k_i f - \sum\limits_{\ell=0}^{i-1}\left(\sum\limits_{j=\ell}^{i-1}(-1)^{j-\ell}{i \choose j}{j \choose \ell}\right)U_\ell^iD^k_\ell f\\
    &= D^k_if - \sum\limits_{\ell=0}^{i-1}(-1)^{i-1-\ell}{i\choose \ell}U_\ell^iD^k_\ell f\\
    &= \sum\limits_{\ell=0}^i (-1)^{i-\ell}{i \choose \ell}U_\ell^iD^k_\ell f.
\end{align*}
The explicit form of $\fup{i}$ then follows simply by applying ${k \choose i}U^k_i$ to $\gup{i}$, and the equivalent form is immediate from the definition of the down and up operators.
\end{proof}

\subsection{Bottom-Up vs. HD-Level-Set}
Dikstein, Dinur, Filmus, and Harsha's HD-Level-Set Decomposition \cite{dikstein2018boolean} is an elegant linear-algebraic decomposition for functions on local-spectral expanders. Like the Bottom-Up Decomposition, it breaks $f \in C_k$ down into $k+1$ Fourier levels, but differs in that it does so in a \textit{top-down} fashion.
\begin{theorem}[HD-Level-Set Decomposition, Theorem 8.2 \cite{dikstein2018boolean}]\label{thm:decomp-ddfh}
Let $(X,\Pi)$ be a $d$-dimensional two-sided $\gamma$-local-spectral expander, $\gamma < \frac{1}{d}$, $0 \leq k \leq d$, and let: 
\[
H^0=C_0, H^i=\text{Ker}(D_i), V_k^i = U^{k}_iH^i.
\]
Then:
\[
C_k = V^0_k \oplus \ldots \oplus V^k_k.
\]
In other words, every $f \in C_k$ has a unique decomposition $f=\sum\limits \fdown{i}$ such that $\fdown{i}=U^{k}_i\gdown{i}$ for $\gdown{i} \in \text{Ker}(D_i)$.
\end{theorem}
While the HD-Level-Set Decomposition is certainly useful in its own right, it has no known explicit form. This can make the analysis of standard Fourier-analytic techniques like restriction difficult, and hampers the analysis of higher moments. We will show that the Bottom-Up Decomposition provides an explicit approximation of the HD-Level-Set Decomposition that circumvents these issues while maintaining the latter's useful properties.

Before jumping into the details, however, it is worth reviewing an elegant technical tool of \cite{dikstein2018boolean} that will be crucial for our analysis. They prove that local-spectral expansion is equivalent to a global notion of spectral expansion on complexes that relates the upper and lower walks.
\begin{theorem}[DDFH Claim 8.8]
Let $(X,\Pi)$ be a $d$-dimensional two-sided $\gamma$-local-spectral expander. Then for any $1 \leq i \leq j \leq d$:
\begin{equation}\label{eq:DDFH}
D_i U_j^i = \frac{j}{i} U_{j-1}^{i-1} D_j + \frac{i-j}{i} U_j^{i-1}+E_{i,j},
\end{equation}
where $\|E_{i,j}\| \le (i-j)\gamma$.
\end{theorem}
It is worth noting that this result (and the HD-Level-Set Decomposition) hold more generally for any ``expanding poset''---the difference lies in the exact coefficients in the above relation. 

The crucial observation for proving that the Bottom-Up Decomposition is a Fourier basis (that explicitly approximates the HD-Level-Set Decomposition) is that while $\gup{i}$ may not lie directly in $\text{Ker}(D_i)$ like $\gdown{i}$, it is fairly close to doing so. Proving this actually relies crucially on the exact coefficients in \Cref{eq:DDFH} which correspond to working over a simplicial complex. As a result, it is not clear the Bottom-Up Decomposition is a Fourier basis at all for general expanding posets.
\begin{lemma}\label{lemma:approx-kernel}
Let $(X,\Pi)$ be a two-sided $\gamma$-local-spectral expander, $f \in C_k$, and $\gup{i}$ be given as in the Bottom-Up Decomposition. Then:
\[
\norm{D_i\gup{i}}_2 \leq 2^{O(i)}\gamma \norm{D^k_i f}_2.
\]
\end{lemma}
\begin{proof}

The result follows from directly expanding $D_i\gup{i}$:
\begin{align*}
    D_i \gup{i} &= \sum_{j=0}^i (-1)^{i-j} {i \choose j} D_i U_j^i D_j^k f\\
    &=\sum_{j=0}^i (-1)^{i-j} {i \choose j}  \left(\frac{j}{i} U_{j-1}^{i-1} D_j + \frac{i-j}{i} U_j^{i-1}+E_{i,j} \right) D_j^k f.
\end{align*}
The key is then to notice that the main terms cancel. That is, setting $c_{i,j} = (-1)^{i-j} \binom{i}{j}$ we have:
\begin{align*}
    \sum_{j=0}^i c_{i,j} \left(\frac{j}{i} U_{j-1}^{i-1} D_j + \frac{i-j}{i} U_j^{i-1} \right) D_j^k f = 
\sum_{j=0}^{i-1} \left(\frac{j+1}{i} c_{i,j+1} + \frac{i-j}{i} c_{i,j} \right)  U_j^{i-1} D_j^k f = 0.
\end{align*}
Finally by the triangle inequality we have:
\begin{align*}
\|D_i \gup{i}\|_2 \le \sum_{j=0}^i \binom{i}{j}
\|E_{i,j} D_j^k f\|_2 &\le \sum_{j=0}^i \binom{i}{j} (i-j) \gamma \|D^k_j f\|_2\\
&\leq O(i 2^i \gamma) \|D^k_if\|_2
\end{align*}
where we have used the fact that $D$ contracts $\ell_2$-norm in the final step.
\end{proof}
We will also need approximate orthogonality of both decompositions.
\begin{lemma}\label{lemma:approx-orthog}
Let $(X,\Pi)$ be a two-sided $\gamma$-local-spectral expander. Then the following three approximate orthogonality relations hold for all $i \neq j$:
\begin{align}
    |\langle \fdown{i}, \fdown{j} \rangle| \leq 2^{O(k)}\gamma \norm{f}^2_2\\
    |\langle \fdown{i}, \fup{j} \rangle| \leq 2^{O(k)}\gamma \norm{f}^2_2\\
    |\langle \fup{i}, \fup{j} \rangle| \leq \min\{k^{O(i+j)},2^{O(k)}\}\gamma \norm{f}^2_2
\end{align}
\end{lemma}
\begin{proof}
All relations follow from \Cref{eq:DDFH} and \Cref{lemma:approx-kernel}.
The first relation is proved in \cite{dikstein2018boolean}. The latter relations follow similarly, but we give the third for completeness. In particular, assuming $i > j$ we have:
\begin{align*}
    |\langle \fup{i}, \fup{j} \rangle| &= \left |{k \choose i}{k \choose j}\langle U^k_i\gup{i}, U^k_j\gup{j} \rangle\right|\\
    &= {k \choose i}{k \choose j}\left|\langle D_{i-1}^kU^{k}_i\gup{i}, U^{i-1}_j\gup{j} \rangle\right|\\
    &\leq {k \choose i}{k \choose j}\norm{D_{i-1}^kU^{k}_i\gup{i}}_2\norm{U^{i-1}_j\gup{j}}_2\\
    &\leq {k \choose i}{k \choose j}\norm{D_{i-1}^kU^{k}_i\gup{i}}_2\norm{\gup{j}}_2
\end{align*}
where we have applied Cauchy-Schwarz and used the fact that averaging operators contract $\ell_2$-norm. We now separately bound both norms. The second is the simpler of the two, and we claim it is at most $2^{O(i)}\norm{f}$. In fact a more general claim holds.
\begin{claim}\label{claim:p-norm}
For any $\ell_p$-norm, we have:
\[
\norm{\gup{i}}_p \leq 2^i \norm{D^k_i f}_p
\]
\end{claim}
\begin{proof}
This follows from direct expansion of the $\ell_p$-norm:
\begin{align*}
\|\gup{i}\|_p &\le \sum_{j=0}^i \binom{i}{j} \|U^i_j D_j^k f\|_p\\
&\le \sum_{j=0}^i \binom{i}{j} \|D^k_if\|_p \\
&= 2^i \|D^k_i f\|_p
\end{align*}
where we have applied the triangle inequality and the fact that $U$ and $D$ contract $p$-norms. 
\end{proof}
Plugging this back into the above, we have:
\[
\langle \fup{i}, \fup{j} \rangle \leq 2^i{k \choose i}{k \choose j}\norm{f}_2\norm{D_{i-1}^kU^{k}_i\gup{i}}_2.
\]
To complete the proof, it is therefore enough to argue that $\norm{D_{i-1}^kU^{k}_i\gup{i}}_2$ is at most $c_k\gamma\norm{f}$ for some small enough $c_k$. This follows from $k-i$ repeated applications of \Cref{eq:DDFH} (one for each instance of the up operator). Informally, each application of \Cref{eq:DDFH} incurs two error terms, one stemming from the matrix $E_{i,j}$, and the other from $\frac{j}{i}U^{i-1}_{j-1}D_i$ applied to $\gup{i}$, which we know is small by \Cref{lemma:approx-kernel}. The final remaining term is then proportional to the original term, but with one $DU$ pair removed. For instance, for the first application we have:
\begin{align*}
\norm{D_{i-1}^kU^{k}_i\gup{i}} &= \norm{ D^{k-1}_{i-1}(D_kU^{k}_i)\gup{i},U^{i-1}_j\gup{j}}\\
&= \norm{D^{k-1}_{i-1}\left(\frac{i}{k} U_{i-1}^{k-1} D_i + \frac{k-i}{k} U_i^{k-1}+E_{k,i}\right)\gup{i}}\\
&= \frac{k-i}{k}\norm{ D^{k-1}_{i-1} U_i^{k-1}\gup{i}} + \frac{i}{k}\norm{D^{k-1}_{i-1}U_{i-1}^{k-1} D_i\gup{i}} + \norm{D^{k-1}_{i-1}E_{k,i}\gup{i}}\\
&\leq \frac{k-i}{k}\norm{ D^{k-1}_{i-1} U_i^{k-1}\gup{i}} + 2^{O(i)}\gamma\norm{f} + k\gamma\norm{\gup{i}}\\
&\leq \frac{k-i}{k}\norm{ D^{k-1}_{i-1} U_i^{k-1}\gup{i}} + k2^{O(i)}\gamma\norm{f}
\end{align*}
where we have used the facts that by \Cref{eq:DDFH} and \Cref{lemma:approx-kernel}, $\norm{E_{k,i}} \leq k\gamma$, and $\norm{D_i\gup{i}} \leq 2^{O(i)}\gamma\norm{f}$. A basic inductive argument then implies that $\norm{D_{i-1}^kU^{k}_i\gup{i}} \leq 2^{O(i)}k^2\gamma\norm{f}$, which completes the proof. For a more formal induction following exactly the same strategy, see \cite{dikstein2018boolean,bafna2020high}.
\end{proof}
With these lemmas in hand, proving that the Bottom-Up Decomposition $\ell_2$-approximates the HD-Level-Set Decomposition is elementary.
\begin{theorem}\label{thm:Bottom-vs-Top}
Let $(X,\Pi)$ be a two-sided $\gamma$-local-spectral expander and $f \in C_k$. Then the Bottom-Up and HD-Level-Set Decomposition are close in $\ell_2$-norm:
\[
\norm{\fup{i} - \fdown{i}}^2_2 \leq 2^{O(k)}\gamma\norm{f}_2^2.
\]
Similarly:
\[
\left |\norm{\fup{i}}^2 - \norm{\fdown{i}}^2 \right | \leq 2^{O(k)}\gamma\norm{f}_2^2
\]
\end{theorem}
\begin{proof}
By \Cref{lemma:approx-orthog}, we have:
\begin{align*}
\langle \fdown{i} - \fup{i}, \fdown{i} \rangle &= \langle f - \fup{i}, \fdown{i} \rangle \pm 2^{O(k)}\gamma\|f\|_2^2\\
&= \pm 2^{O(k)}\gamma\|f\|_2^2
\end{align*}
Similarly:
\begin{align*}
\langle \fdown{i} - \fup{i}, \fup{i} \rangle &=\pm 2^{O(k)}\gamma\|f\|_2^2
\end{align*}
Therefore:
\begin{align*}
\langle \fdown{i} - \fup{i}, \fdown{i} - \fup{i} \rangle \leq 2^{O(k)}\gamma\|f\|_2^2
\end{align*}
as desired. To prove the second inequality, note that:
\begin{align*}
    \left |\langle \fup{i}, \fup{i} \rangle - \langle \fdown{i},\fdown{i} \rangle \right | &= \left |\langle \fup{i} - \fdown{i}, \fup{i} - \fdown{i} \rangle + 2\langle \fdown{i}, \fup{i} - \fdown{i} \rangle \right |\\
    &\leq 2^{O(k)}\gamma \norm{f}^2
\end{align*}
by our previous observations.
\end{proof} 
Note that since the HD-Level-Set Decomposition satisfies the Fourier-anatlyic properties in \Cref{thm:intro-BU-properties} \cite{dikstein2018boolean,alev2019approximating,bafna2020high}, \Cref{thm:Bottom-vs-Top} implies that the Bottom-Up Decomposition does as well.

\subsection{Properties of the Bottom-Up Decomposition}
Our proof of hypercontractivity (\Cref{thm:intro-hypercontraction}) relies on a number of important structural properties of and relations between $g_i$ and $f_i$. The first (and most basic) of these is the analog of a classic result for the HD-Level-Set Decomposition relating to $\ell_2$-norms of $g_i$ and $f_i$. 
\begin{lemma}\label{lemma:g-vs-f}
Let $(X,\Pi)$ be a two-sided $\gamma$-local-spectral expander and $f \in C_k$. Then:
\[
\langle \gup{i}, \gup{i} \rangle = \frac{1}{{k \choose i}}\langle \fup{i}, \fup{i} \rangle \pm c_{k,i}\gamma\norm{D^k_i f}_2^2
\]
where $c_{k,i} \leq k^{O(i)}$.
\end{lemma}
\begin{proof}
The proof is essentially the same as for the HD-Level-Set Decomposition and as our analysis above, though we repeat the idea for completeness. The key is again to apply \Cref{eq:DDFH}. In particular, recall that:
\begin{align*}
    \langle \fup{i}, \fup{i} \rangle &= {k \choose i}^2\langle U^k_i\gup{i},U^k_i\gup{i}\rangle\\
    &={k \choose i}^2\langle \gup{i},D^k_iU^k_i\gup{i}\rangle
\end{align*}
by adjointness of $D$ and $U$ \cite{dikstein2018boolean}. The proof is then essentially the same as \Cref{lemma:approx-orthog}. Repeated application of \Cref{eq:DDFH} gives an error term of $O(k^2)\gamma\norm{D^k_i f}_2$. The only difference is that the main term no longer has an extra occurrence of $D$ at the end. Thus instead of becoming another error term, the main term becomes:
\[
{k \choose i}^2\left(\prod\limits_{j=0}^{k-i-1} \frac{k-j-i}{k-j}\right)\langle \gup{i},\gup{i} \rangle = {k \choose i}\langle \gup{i},\gup{i} \rangle
\]
which gives the result.
\end{proof}
We now cover a few important bounds on $g_i$ for \textit{pseudorandom functions}.
\begin{definition}[Pseudorandom]
Let $(X,\Pi)$ be a simplicial complex. We say that $f \in C_k$ is $(\eps,i)$-pseudorandom if it is sparse across all $i$-links in two senses:
\begin{enumerate}
    \item For all $\tau \in X(i)$: 
    \[
    \left| \underset{X_\tau}{\mathbb{E}}[f] \right| \leq \varepsilon \norm{f}_\infty
    \]
    \item For all $\tau \in X(i)$: 
    \[
    \langle f|_\tau, f|_\tau \rangle \leq \varepsilon \norm{f}_\infty^2
    \]
\end{enumerate}
\end{definition}
We note that if $f$ is non-negative, the former condition implies the latter: 
\begin{align*}
\langle \restrict{f}{\tau}, \restrict{f}{\tau} \rangle &\leq \norm{\restrict{f}{\tau}}_1\norm{f}_\infty\\ 
&= \mathbb{E}[f|_\tau]\norm{f}_\infty\\
&\leq \varepsilon\norm{f}_\infty^2.
\end{align*}
It is also worth noting that any $(\eps,i)$-pseudorandom function is automatically $(\eps,j)$-pseudorandom for $j \leq i$.

We now cover the first property of the Bottom-Up Decomposition that does not follow from standard HDX analysis, the behavior of level functions under restriction. Analysis of restrictions is a classic Fourier analytic tool, and the fact that our decomposition behaves nicely under restriction is a major advantage over previous decompositions which have no clear local structure in this sense. For this particular work, we'll mostly be interested in the following bound on the $\ell_2$-norm of restrictions.
\begin{proposition}\label{lemma:g-link-2nd-moment}
Let $(X,\Pi)$ be a two-sided $\gamma$-local-spectral expander and $f \in C_k$ an $(\eps,i)$-pseudorandom function. Then for any $j \leq i \leq k$ and $\tau \in X(j)$:
\[
\langle \gup{i}|_{\tau}, \gup{i}|_{\tau} \rangle \leq \left(\frac{\eps}{{k-j \choose i-j}} + c_{k,i}\gamma\right)\norm{f}^2_\infty,
\]
where $c_{k,i} \leq k^{O(i)}$.
\end{proposition}
Proving this, however, requires a more general understanding of the Bottom-Up Decomposition under restriction. They key observation is that the restriction of our level functions is closely related to the level functions of the restriction. More formally, for any $\tau \in X(j)$ let $\gup{\ell}^{(\tau)}$ denote the Bottom-Up Decomposition of $f|_\tau$. Then following relation between $g_i|_\tau$ and the $\gup{\ell}^{(\tau)}$ holds.
\begin{lemma}\label{lemma:g-restriction}
Let $(X,\Pi)$ be a pure, weighted simplicial complex, and $f \in C_k$. Then for any $j \leq i \leq k$ and $\tau \in X(j)$:
\[
\gup{i}|_\tau = \sum\limits_{\sigma \subseteq \tau}(-1)^{|\sigma|}\gup{i-j}^{(\tau \setminus \sigma)}.
\]
\end{lemma}
\begin{proof}
This follows almost immediately from directly expanding the definition of $\gup{i}|_\tau$. In particular, recall that for all $I \in X(i-j)$, we have by \Cref{thm:explicit}:
\[
\restrict{\gup{i}}{\tau}(I) = \sum\limits_{T \subseteq I \cup \tau} (-1)^{|(\tau \cup I) \setminus T|} \underset{X_T}{\mathbb{E}}[f].
\]
The trick is to notice that we can divide up this sum over $T \subseteq I \cup \tau$ by $T$'s intersection with $\tau$. It will be convenient to phrase this in the following way. Let $\mathscr{T}$ denote the set of all sub-faces $T \subseteq I \cup \tau$, and for each $\sigma \subset \tau$, let $\mathscr{T}_\sigma$ be the set of sub-faces $T \subset I \cup \tau$ such that $T \cap \tau = \tau \setminus \sigma$. Notice that for any $\sigma \neq \sigma'$, $\mathscr{T}_\sigma$ and $\mathscr{T}_{\sigma'}$ are disjoint, and that the union of these families is exactly $\mathscr{T}$. Together, this means that we can break up the above sum by first summing over $\sigma$, and then every $T \in \mathscr{T}_\sigma$:
\[
\restrict{\gup{i}}{\tau}(I) = \sum\limits_{\sigma \subseteq \tau} \left(\sum\limits_{T \in \mathscr{T}_\sigma } (-1)^{|(I \cup \tau) \setminus T|}\underset{X_{T}}{\mathbb{E}}[f]\right)
\]
By definition, every $T \in \mathscr{T}_\sigma$ can be written as $T' \cup (\tau \setminus \sigma)$. Plugging this into the above gives the result:
\begin{align*}
    \sum\limits_{\sigma \subseteq \tau} \left(\sum\limits_{T \in \mathscr{T}_\sigma } (-1)^{|(I \cup \tau) \setminus T|}\underset{X_{T}}{\mathbb{E}}[f]\right) &= \sum\limits_{\sigma \subseteq \tau} \left(\sum\limits_{T' \cup (\tau \setminus \sigma) \in \mathscr{T}_\sigma } (-1)^{|(I \cup \tau) \setminus (T' \cup (\tau \setminus \sigma))|}\underset{X_{T' \cup (\tau \setminus \sigma)}}{\mathbb{E}}[f]\right)\\
    &= \sum\limits_{\sigma \subseteq \tau} (-1)^{|\sigma|}\left(\sum\limits_{T' \cup (\tau \setminus \sigma) \in \mathscr{T}_\sigma } (-1)^{|I \setminus T'|}\underset{X_{T' \cup (\tau \setminus \sigma)}}{\mathbb{E}}[f]\right)\\
    &= \sum\limits_{\sigma \subseteq \tau} (-1)^{|\sigma|} \gup{i-j}^{(\tau \setminus \sigma)}(I),
\end{align*}
where the final step comes from the fact that the inner summation over $\mathscr{T}_\sigma$ is equivalent to summing over all $T'$ in the link of $\tau \setminus \sigma$.
\end{proof}
We note that the same result was known to hold for the Bottom-Up Decomposition over the complete complex \cite{khot2018small}, who proved the result by induction using the recursive form of the decomposition. The same strategy will work for general simplicial complexes, but we find using the explicit form as above to be a bit simpler.

With this in hand, proving \Cref{lemma:g-link-2nd-moment} is fairly elementary and follows similarly to its analogous statement for the complete complex (see \cite[Corollary 3.4]{khot2018small}).
\begin{proof}[Proof of \Cref{lemma:g-link-2nd-moment}]
An application of \Cref{lemma:g-restriction} and Cauchy-Schwarz implies that:
\begin{align*}
\gup{i}|_\tau(T)^2 &= \left(\sum\limits_{\sigma \subseteq \tau}(-1)^{|\sigma|}\gup{i-j}^{(\tau \setminus \sigma)}(T)\right)^2\\
&\leq 2^{O(i)}\sum\limits_{\sigma \subseteq \tau}\gup{i-j}^{(\tau \setminus \sigma)}(T)^2
\end{align*}
Then applying \Cref{lemma:g-vs-f} gives:
\begin{align*}
    \langle \gup{i}|_\tau, \gup{i}|_\tau \rangle &\leq 2^{O(i)}\sum\limits_{\sigma \subseteq \tau} \langle \gup{i-j}^{(\tau \setminus \sigma)}, \gup{i-j}^{(\tau \setminus \sigma)} \rangle\\
    &\leq 2^{O(i)}\sum\limits_{\sigma \subseteq \tau} \frac{\langle \restrict{f}{(\tau \setminus \sigma)},\restrict{f}{(\tau \setminus \sigma)}\rangle}{{k - j + |\sigma| \choose i-j}} + c_1\gamma\norm{f}_\infty^2\\
    &\leq 2^{O(i)}\sum\limits_{\sigma \subseteq \tau} \frac{\varepsilon\norm{f}_\infty^2}{{k - j + |\sigma| \choose i-j}} + c_1\gamma\norm{f}_\infty^2\\
     &\leq 2^{O(i)}\frac{\varepsilon\norm{f}_\infty^2}{{k - j \choose i-j}} + c_2\gamma\norm{f}_\infty^2
\end{align*}
where $c_1,c_2 \leq k^{O(i)}$.
\end{proof}
Finally, it will also be useful to bound the infinity norm of $\gup{i}$ as well. Our final property shows that $\norm{\gup{i}}_{\infty}$ is particularly small when $f$ is pseudorandom.
\begin{lemma}\label{lemma:g-infty}
Let $(X,\Pi)$ be a two-sided $\gamma$-local-spectral expander and $f \in C_k$ be any $(\eps,i)$-pseudorandom function satisfying $\E[f] \leq \eps\norm{f}_\infty$. Then the infinity norm of $\gup{i}$ is small:
\[
\norm{\gup{i}}_{\infty} \leq 2^{i}\eps\norm{f}_\infty
\]
\end{lemma}
\begin{proof}
This is immediate from combining the explicit form of $\gup{i}$ with $(\eps,i)$-pseudorandomness.
\begin{align*}
    |\gup{i}(w)| &= \left|\sum_{j=0}^i (-1)^{i-j} \binom{i}{j} U_j^i D_j^k f(w)\right|\\
    &\leq \sum_{j=0}^i \binom{i}{j} \eps\norm{f}_\infty\\
    &= 2^{i} \eps\norm{f}_\infty
\end{align*}
where we have used the observation that since $f$ is $(\eps,i)$-pseudorandom, for all $w \in X(i)$: 
\[
|U_j^i D_j^k f(w)| \leq \eps\norm{f}_\infty
\]
\end{proof}

\section{Hypercontractivity on HDX}\label{sec:hyper}
In this section, we prove a hypercontractivity theorem for the Bottom-Up Decomposition on two-sided local-spectral expanders. Since we will work only with the Bottom-Up Decomposition in this section, we drop the $\uparrow$ for simplicity and simply write
$f = \sum f_i$ for $f_i = {k \choose i}U^k_ig_i$ and $g_i=\gup{i}$ as defined in the Bottom-Up Decomposition. 

\begin{theorem}\label{thm:hypercontraction}
Let $(X,\Pi)$ be a two-sided $\gamma$-local-spectral expander with $\gamma \leq k^{-\Omega(i)}$, and $f \in C_k$ an $(\eps,i)$-pseudorandom function. If $f=f_0+\ldots+f_k$ is the Bottom-Up Decomposition of $f$, then:
\[
\E[f_i^4] \leq 2^{O(i)}\eps\E[f_i^2]\norm{f}_\infty^2 + c_{k,i}\varepsilon\gamma^{1/2}\norm{D^k_if}_2^2\norm{f}_\infty^2
\]
where $c_{k,i} \leq k^{O(i)}$.\footnote{Note that this can be improved to $c_{k,i} \leq \max\left\{2^{O(i)},{k \choose i}^{O(1)}\right\}$, but since we generally consider the regime of $i \ll k$ we use $k^{O(i)}$ throughout for simplicity.}
\end{theorem}
The proof of \Cref{thm:hypercontraction} can get a bit technical at points, so for simplicity of notation, we note it is sufficient to prove the result assuming $\norm{f}_\infty=1$. Given a general function $f$, applying this to $\frac{f}{\norm{f}_\infty}$ gives the general form in \Cref{thm:hypercontraction}. Keeping this in mind, we'll start by laying out our general strategy for analyzing the fourth moment. Let $[\tau]_i = \{a \subseteq \tau: a \in X(i)\}$, and note that $f_i(\tau) = \sum_{a \in [\tau]_i} g_i(a)$. Using this notation, we can expand out the $4$th moment of $f_i$:
\[
\E[f_i^4] = \underset{\tau \in X(k)}{\E} \sum_{a,b,c,d \in [\tau]_i} g_i(a) g_i(b) g_i(c) g_i(d) = \sum_{a,b,c,d \in X(i)} \pi_k(X_{a \cup b \cup c \cup d}) g_i(a) g_i(b) g_i(c) g_i(d),
\]
where the indices $a,b,c,d$ are \textit{ordered}. We can further simplify this by grouping the terms by size of $a \cup b \cup c \cup d$:
\[
\E[f_i^4] = \sum_{\ell=i}^{4i}{k \choose \ell} \sum_{e \in X(\ell)} \pi_\ell(e) \sum_{a,b,c,d \in X(i): a \cup b \cup c \cup d = e} g_i(a) g_i(b) g_i(c) g_i(d).
\]
Analyzing the RHS directly is difficult, so taking after \cite{khot2018small}, we will partition the term even further by summing over fixed \textbf{intersection patterns} of $a$, $b$, $c$, and $d$ (an intersection pattern fixes the intersection size of every subset of $\{a,b,c,d\}$). Denote the set of such patterns where $|a \cup b \cup c \cup d| = \ell$ by $\Sigma_\ell$, and for any $e \in X(\ell)$, and $\sigma \in \Sigma_\ell$, let $\sigma(e)$ denote all tuples $(a,b,c,d)$ such that $a \cup b \cup c \cup d = e$, $(a,b,c,d) \in \sigma$. We may now write:
\[
\E[f_i^4] = \sum_{\ell=i}^{4i}{k \choose \ell}\sum_{\sigma \in \Sigma_\ell} \sum_{e \in X(\ell)} \pi_\ell(e) \sum_{a,b,c,d \in \sigma(e)} g_i(a) g_i(b) g_i(c) g_i(d).
\]
We make one final simplification of the above before moving to analysis. Let $x_1,\ldots,x_\ell$ be random variables which take on vertex values in the complex. For each intersection pattern $\sigma \in \Sigma_\ell$, let $I^{\sigma}_1,\ldots,I^{\sigma}_4$ be size-$i$ subsets of $\{x_1,\ldots,x_\ell\}$ whose union is $\{x_1,\ldots,x_\ell\}$ and which satisfy the intersection pattern $\sigma$. Then we can simplify the above as the following expectation over the $x_i$:
\[
\E[f_i^4] = \sum_{\ell=i}^{4i}{k \choose \ell}\sum_{\sigma \in \Sigma_\ell} \beta(\sigma) \underset{x_1 \in X(1)}{\E}\left[\underset{x_2 \in X_{x_1}(1)}{\E}\ldots\left[\underset{x_\ell \in X_{x_1,\ldots,x_{\ell-1}}(1)}{\E}\left[g_i(I^{\sigma}_1)g_i(I^{\sigma}_2)g_i(I^{\sigma}_3)g_i(I^{\sigma}_4) \right]\right]\right].
\]
where $\beta(\sigma) \leq 2^{O(i)}$ is a parameter dependent on the intersection pattern that accounts for the new normalization of terms in the nested expectation.\footnote{The bound $2^{O(i)}$ follows from a simple counting argument. Fix an intersection pattern $\sigma=\{a_1,\ldots,a_{15}\}$. Expanding out the expectation $\underset{x_1 \in X(1)}{\E}\left[\underset{x_2 \in X_{x_1}(1)}{\E}\ldots\left[\underset{x_\ell \in X_{x_1,\ldots,x_{\ell-1}}(1)}{\E}\left[g_i(I^{\sigma}_1)g_i(I^{\sigma}_2)g_i(I^{\sigma}_3)g_i(I^{\sigma}_4) \right]\right]\right]$. The coefficient of any term $g_i(a)g_i(b)g_i(c)g_i(d)$ is given by $\frac{\pi_\ell(e)}{\ell!}$ times the number of permutations of $x_1,\ldots x_\ell$ that fix $g_i(a)g_i(b)g_i(c)g_i(d)$. This latter value is the same for any term, and can be lower bounded by $\prod_{i=1}^{15} a_i!$, so the final coefficient is at least $\pi_\ell(e)\frac{\prod a_i!}{\ell!}$. To cancel this value, $\beta(\sigma)$ is therefore at most $\frac{\ell!}{\prod a_i!} \leq 2^{O(i)}$ as desired.} For simplicity of notation, we will instead write the right-hand side as:
\[
\E[f_i^4] = \sum_{\ell=i}^{4i}{k \choose \ell}\sum_{\sigma \in \Sigma_\ell} \beta(\sigma) \underset{x_1,\ldots,x_\ell}{\E}\left[g_i(I^{\sigma}_1)g_i(I^{\sigma}_2)g_i(I^{\sigma}_3)g_i(I^{\sigma}_4) \right]
\]
where it is understood that $\underset{x_1,\ldots,x_\ell}{\E}$ is a shorthand for the nested expectation $\underset{x_1 \in X(1)}{\E}\underset{x_2 \in X_{x_1}(1)}{\E}\ldots\underset{x_\ell \in X_{x_1,\ldots,x_{\ell-1}}(1)}{\E}$. We will use this convention throughout the rest of the proof, as the nested notation is cumbersome to write otherwise.

Our goal is now to upper bound this sum to get a hypercontractive inequality. We do this by bounding each sign pattern independently.
\begin{claim}\label{claim:sign-pattern-bound}
For every sign pattern $\sigma$, the corresponding expectation is bounded by:
\[ 
\underset{x_1,\ldots,x_\ell}{\E}\left[g_i(I^{\sigma}_1)g_i(I^{\sigma}_2)g_i(I^{\sigma}_3)g_i(I^{\sigma}_4) \right] \leq \left(\frac{i}{k} \right)^{\ell}2^{O(i)}\eps\E[f_i^2] + c_{k,i}\varepsilon\gamma^{1/2}\norm{D_i^k f}_2^2
\]
where $c_{k,i} \leq k^{O(i)}$.
\end{claim}
Before jumping into the proof of \Cref{claim:sign-pattern-bound}, let's show how it can be used to prove \Cref{thm:hypercontraction}.
\begin{proof}[Proof of \Cref{thm:hypercontraction}]
Recall it is sufficient to prove the result assuming $\norm{f}_\infty=1$. As discussed earlier in the section, expanding the 4th moment gives the following relation:
\[
\E[f_i^4] = \sum_{\ell=i}^{4i}{k \choose \ell}\sum_{\sigma} \beta(\sigma) \underset{x_1,\ldots,x_\ell}{\E}\left[g_i(I^{\sigma}_1)g_i(I^{\sigma}_2)g_i(I^{\sigma}_3)g_i(I^{\sigma}_4) \right].
\]
Applying \Cref{claim:sign-pattern-bound} to the righthand side gives:
\begin{align*}
    \E[f_i^4] & \leq \sum_{\ell=i}^{4i}{k \choose \ell}\sum_{\sigma \in \Sigma_\ell} \beta(\sigma) \left (\left(\frac{i}{k} \right)^{\ell}2^{O(i)}\eps\E[f_i^2] + c_{k,i}\varepsilon\gamma^{1/2}\norm{D^k_i f}_2^2\right)\\
    &\leq \sum_{\ell=i}^{4i}\left(\frac{ek}{\ell}\right)^\ell\left(\frac{i}{k} \right)^{\ell}\left(\sum_{\sigma \in \Sigma_\ell} \beta(\sigma)\right) \left (2^{O(i)}\eps\E[f_i^2] + c_1\varepsilon\gamma^{1/2}\norm{D^k_i f}_2^2\right)\\
    &\leq 2^{O(i)}\eps \E[f_i^2] + c_2\varepsilon\gamma^{1/2}\norm{D^k_i f}_2^2
\end{align*}
where $c_1,c_2 \leq k^{O(i)}$ and the last step follows from noting that there are at most $\text{poly}(i)$ intersection patterns.
\end{proof}
\subsection{Proving \Cref{claim:sign-pattern-bound}}
The main technical work comes in proving \Cref{claim:sign-pattern-bound}, which relies heavily on Garland's method and our new localization strategy for decorrelating variables (\Cref{cor:localization}).

We split the proof into two parts. First, we will show that any pattern which has a unique element (i.e. some $x_i$ which appears only in one of the four sets) may be disregarded. 
\begin{proposition}
If $\sigma$ is a pattern in which any variable is unique (appears in only one $I_j$), then:
\[
\underset{x_1,\ldots,x_\ell}{\E}\left[g_i(I^{\sigma}_1)g_i(I^{\sigma}_2)g_i(I^{\sigma}_3)g_i(I^{\sigma}_4) \right] \leq 2^{O(i)} \gamma \eps^2\norm{D^k_i f}_2^2.
\]
\end{proposition}
\begin{proof}
To simplify notations in the proof, let $I=\{x_1,\ldots,x_{\ell}\}$ and $I_j = I^{\sigma}_j$.
Assume without loss of generality that $I_4$ has a unique variable $x_\ell$, and set $J = I \setminus \{x_{\ell}\}$ and $J_4 = I_4 \setminus \{x_{\ell}\}$. We can re-write our expectation as:
\[
(*) = \underset{J}{\E}\left[g_i(I_1)g_i(I_2)g_i(I_3)\underset{x_\ell \in X_{J}(1)}{\E}\left[g_i|_{J_4}(x_\ell)\right] \right].
\]
By \Cref{cor:localization}, the inner expectation can be replaced with $D_ig_i(J_4)$ up to $\gamma$ error in the following sense. Consider any fixing of the variables $J$ (namely, fixing $x_1,\ldots,x_{\ell-1}$), we have:
\begin{align*}
\underset{x_\ell \in X_J(1)}{\E}\left[g_i|_{J_4}(x_\ell)\right] &= \underset{x_\ell \in X_{J_4}(1)}{\E}\left[g_i|_{J_4}(x_\ell)\right] + \Gamma \restrict{g_i}{J_4}(J \setminus J_4)\\
&= D_ig_i(J_4) + \Gamma \restrict{g_i}{J_4}(I \setminus I_4)
\end{align*}
where $\norm{\Gamma} \le O(i \gamma)$ by \Cref{cor:localization}. Plugging this back into our original expectation gives:
\begin{align*}
(*) = \E_{J}\left[g_i(I_1)g_i(I_2)g_i(I_3)D_ig_i(J_4) \right] + \E_{J}\left[g_i(I_1)g_i(I_2)g_i(I_3) \Gamma \restrict{g_i}{J_4} (I \setminus I_4) \right] 
\end{align*}
The idea is now to split each term into two parts: the first three terms $g_i(I_1)g_i(I_2)g_i(I_3)$ and the last term. Let's first split these portions by Cauchy-Schwarz to get:
\begin{align*}
(*) \le \; & \E_{J}\left[g_i(I_1)^2g_i(I_2)^2g_i(I_3)^2\right]^{1/2}\\
\cdot & \left( \underset{J}{\E}\left[ D_i g_i(J_4)^2 \right]^{1/2} + \underset{J_4}{\E}\left[\underset{I \setminus I_4 \in X_{J_4}}{\E}\left[ \Gamma \restrict{g_i}{J_4}(I \setminus I_4)^2 \right]\right]^{1/2}\right)
\end{align*}
where we have re-arranged variable for convenience in the last term. We now bound each term separately. 

The first term can be bounded by the observation that $\norm{g_i}_{\infty} \leq 2^{O(i)}\eps\norm{f}_\infty$ (\Cref{lemma:g-infty}), and hence:
\[
\E_{J}\left[g_i(I_1)^2g_i(I_2)^2g_i(I_3)^2\right]^{1/2} \leq 2^{O(i)}\eps^2\norm{g_i}_2
\]
where we simply bounded two of the three $g_i^2$ terms by their infinity norm and applied Garland's lemma for localizations (\Cref{lemma:garland-local}) to remove the extra variables. 

We next analyze the second term. The first summand is exactly $\norm{D_ig_i}$, which by \Cref{lemma:approx-kernel} is at most $O(\gamma\norm{D^k_i f})$. The second summand is a bit trickier, but can be analyzed through a combination of standard spectral bounds and Garland's lemma for restrictions (\Cref{lemma:garland-restrict}). In particular, re-writing the inner expectation as an inner-product we get:
\begin{align*}
\underset{J_4}{\E}\left[\underset{I \setminus I_4 \in X_{J_4}}{\E}\left[ \Gamma \restrict{g_i}{J_4}(I \setminus I_4)^2 \right]\right]^{1/2} &= \underset{J_4}{\E}\left[\left\langle \Gamma \restrict{g_i}{J_4}, \Gamma \restrict{g_i}{J_4}\right\rangle\right]^{1/2}\\
&\leq c\gamma \underset{J_4}{\E}\left[\left\langle \restrict{g_i}{J_4}, \restrict{g_i}{J_4}\right\rangle\right]^{1/2}\\
&= c\gamma\norm{g_i}_2,
\end{align*}
where $c \leq O(i)$ and we have applied the fact that $\|\Gamma\| \le O(i \gamma)$ and Garland's lemma for restrictions (\Cref{lemma:garland-restrict}). Recalling from \Cref{claim:p-norm} that $\norm{g_i}_2 \leq 2^i\norm{D^k_i f}_2$ completes the result.
\end{proof}
We may now restrict our analysis to patterns in which every variable appears at least twice. Note that this implies $\ell \leq 2i$, which is important because we expect our expectation to scale at best with $k^{-2i}$, so any terms with $\ell > 2i$ would cause difficulty. As in \cite{khot2018small}, we break this analysis into two steps. Let $I_1,\ldots,I_4$ satisfy intersection pattern $\sigma$ as above (we drop the $\sigma$ superscript for convenience), and let $H_i$ for $i \in \{2,3,4\}$ denote the set of variables that appear $i$ times.

We'll start by handling $H_2$ through a combination of Cauchy-Schwarz, Garland's method, and our localization technique for higher moments. Unlike the case of the complete complex studied in \cite{khot2018small}, these latter components are necessary due to the fact that local-spectral expanders are generally far from product spaces (a crucial property of the complete complex exploited in \cite{khot2018small}). The proof is fairly technical, so we'll start by laying out some convenient notation. For any $0 \leq m \leq \ell$, let $T^m=\{x_1,\ldots,x_m\}$. Let $j=|H_3 \cup H_4|$ where $0 \le j \le \ell$.
Noting that re-ordering the variables $x_1,\ldots,x_\ell$ has no effect on the distribution,
we may assume without loss of generality that $H_3 \cup H_4=\{x_1,\ldots,x_j\}$ (where if $j=0$ then $H_3 \cup H_4$ is empty).
Finally, we introduce two useful notations: for $m \le \ell$ let $I_r^m = I_r \cap \{x_1,\ldots,x_m\}$ and $s_r^m = i - |I_r^m|$.

\begin{proposition}\label{claim:H_2}
\begin{align*}
\underset{x_1,\ldots,x_\ell}{\E}\left[g_i(I_1)g_i(I_2)g_i(I_3)g_i(I_4) \right] &\leq
\underset{x_1,\ldots,x_j}{\E}\left[\sqrt{\prod\limits_{r=1}^4\underset{{\tau_r \sim X_{T^j}(s_r^j)}}{\E}[g_i^2|_{I_r^j}(\tau_r)]} \right] + 2^{O(i)}\gamma^{1/2}\eps^2\norm{g_i}_2^2
\end{align*}
\end{proposition}
\begin{proof}
The proof follows from an inductive argument where we pull one variable in $x \in H_2$ inside the sum in each step by de-correlating the two copies of $g_i$ which do not take $x$ as an input, and then applying Cauchy-Schwarz. In particular, we will show by induction that for all $\ell \geq m \geq j$:
\begin{align*}
\underset{x_1,\ldots,x_\ell}{\E}\left[g_i(I_1)g_i(I_2)g_i(I_3)g_i(I_4) \right] \leq \underset{x_1,\ldots,x_m}{\E}\left[\sqrt{\prod\limits_{r=1}^4\underset{{\tau_r \sim X_{T^m}(s^m_r)}}{\E}[g_i^2|_{I_r^m}(\tau_r)]} \right] + 2^{O(i)}\gamma^{1/2}\eps^2 \norm{g_i}_2^2
\end{align*}
The base case ($m=\ell$) is trivial. Since we also done if $m=j$, we may assume that $x_m \in H_2$ and therefore lies in exactly two of $I_1,I_2,I_3,I_4$ by definition. Assume without loss of generality that $x_m \in I_3,I_4$. We'd like to pull $x_m$ inside the expectation. The issue is that despite the fact that $x_m$ does not participate in $I_1$ or $I_2$, these terms actually depend on $x_m$ regardless since $\tau_1$ and $\tau_2$ are drawn from a link that includes $x_m$. To fix this, we can use \Cref{cor:localization} to de-correlate these terms from $x_m$:
\begin{align*} 
\sqrt{\prod\limits_{r=1}^2\underset{{\tau_r \sim X_{T^m}(s^m_r)}}{\E}[g_i^2|_{I_r^m}(\tau_r)]} &= \sqrt{\prod\limits_{r=1}^2 \left( \underset{{\tau_r \sim X_{T^{m-1}}(s^{m-1}_r)}}{\E}\left[g_i^2|_{I_r^{m-1}}(\tau_r)\right] + \Gamma g_i^2|_{I_r^{m-1}}(x_m)\right)}\\
&\leq \sqrt{\underset{{\tau_1 \sim X_{T^{m-1}}(s^{m-1}_1)}}{\E}\left[g_i^2|_{I_1^{m-1}}(\tau_1)\right]}\sqrt{\underset{{\tau_2 \sim X_{T^{m-1}}(s^{m-1}_2)}}{\E}\left[g_i^2|_{I_2^{m-1}}(\tau_1)\right]}\\
&+ \sqrt{\Gamma g_i^2|_{I_2^{m-1}}(x_m)}\sqrt{\underset{{\tau_1 \sim X_{T^{m-1}}(s^{m-1}_1)}}{\E}\left[g_i^2|_{I_1^{m-1}}(\tau_1)\right]}\\
&+ \sqrt{\Gamma g_i^2|_{I_1^{m-1}}(x_m)}\sqrt{\underset{{\tau_2 \sim X_{T^{m-1}}(s^{m-1}_2)}}{\E}\left[g_i^2|_{I_2^{m-1}}(\tau_2)\right]}\\
&+ \sqrt{\Gamma g_i^2|_{I_1^{m-1}}(x_m)}\sqrt{\Gamma g_i^2|_{I_2^{m-1}}(x_m)}\\
\end{align*}
where $\norm{\Gamma} \leq O(i\gamma)$ and we have used the fact that by assumption $I_r^m = I_r^{m-1}$ for $r=1,2$. For the moment, denote the last 3 terms by $err(g)$. Then by the inductive hypothesis we have:
\begin{align*}
&\underset{x_1,\ldots,x_\ell}{\E}\left[g_i(I_1)g_i(I_2)g_i(I_3)g_i(I_4) \right]\\
\leq & \underset{x_1,\ldots,x_m}{\E}\left[\sqrt{\prod\limits_{r=1}^4\underset{{\tau_r \sim X_{T^m}(s_r)}}{\E}[g_i^2|_{I_r^m}(\tau_r)]} \right] + 2^{O(i)}\gamma^{1/2}\eps^2 \norm{g_i}_2^2\\
\leq &\underset{x_1,\ldots,x_{m-1}}{\E}\left[\sqrt{\prod\limits_{r=1}^2\underset{{\tau_r \sim X_{T^{m-1}}(s_r)}}{\E}[g_i^2|_{I_r^{m-1}}(\tau_r)]}\underset{x_m}{\E}\left[\sqrt{\underset{{\tau_3 \sim X_{T^m}(s_3)}}{\E}[g_i^2|_{I_3^m}(\tau_3)]}\sqrt{\underset{{\tau_4 \sim X_{T^m}(s_4)}}{\E}[g_i^2|_{I_4^m}(\tau_4)]} \right] \right]\\
+ &\underset{x_1,\ldots,x_{m}}{\E}\left[err(g)\sqrt{\underset{{\tau_3 \sim X_{T^m}(s^m_3)}}{\E}[g_i^2|_{I_3^m}(\tau_3)]}\sqrt{\underset{{\tau_4 \sim X_{T^m}(s^m_4)}}{\E}[g_i^2|_{I_4^m}(\tau_4)]} \right] + 2^{O(i)}\gamma^{1/2}\eps^2 \norm{g_i}_2^2
\end{align*}
By Cauchy-Schwarz, the first term can be bounded by:
\[
\underset{x_1,\ldots,x_{m-1}}{\E}\left[\sqrt{\prod\limits_{r=1}^4\underset{{\tau_r \sim X_{T^{m-1}}(s^{m-1}_r)}}{\E}[g_i^2|_{I_r^{m-1}}(\tau_r)]} \right],
\]
so it is enough to show that the latter error term is small. We'll analyze each term in $err(g)$ independently using Cauchy-Schwarz, Garland's method, and our bound on $\norm{g}_{\infty}$. Starting with the first term, an application of Cauchy-Schwarz gives:
\begin{align*}
    &\underset{x_1,\ldots,x_{m}}{\E}\left[\sqrt{\Gamma g_i^2|_{I_2^{m-1}}(x_m)}\sqrt{\underset{{\tau_1 \sim X_{T^{m-1}}(s^{m-1}_1)}}{\E}\left[g_i^2|_{I_1^{m-1}}(\tau_1)\right]}\sqrt{\underset{{\tau_3 \sim X_{T^m}(s^m_3)}}{\E}[g_i^2|_{I_3^m}(\tau_3)]}\sqrt{\underset{{\tau_4 \sim X_{T^m}(s^m_4)}}{\E}[g_i^2|_{I_4^m}(\tau_4)]} \right]\\ 
    \leq&  \underset{x_1,\ldots,x_{m}}{\E}\left[\Gamma g_i^2|_{I_2^{m-1}}(x_m)\underset{{\tau_1 \sim X_{T^{m-1}}(s^{m-1}_1)}}{\E}\left[g_i^2|_{I_1^{m-1}}(\tau_1)\right]\right]^{1/2}\underset{x_1,\ldots,x_{m}}{\E}\left[\underset{{\tau_3 \sim X_{T^m}(s^m_3)}}{\E}[g_i^2|_{I_3^m}(\tau_3)]\underset{{\tau_4 \sim X_{T^m}(s^m_4)}}{\E}[g_i^2|_{I_4^m}(\tau_4)] \right]^{1/2}.
    \end{align*}
The righthand expectation is easy to analyze using the fact that $\norm{g_i}_{\infty} \leq 2^{O(i)}\eps$:
\begin{align*}
    \underset{x_1,\ldots,x_{m}}{\E}\left[\underset{{\tau_3 \sim X_{T^m}(s^m_3)}}{\E}[g_i^2|_{I_3^m}(\tau_3)]\underset{{\tau_4 \sim X_{T^m}(s^m_4)}}{\E}[g_i^2|_{I_4^m}(\tau_4)] \right]^{1/2}
    &\leq 2^{O(i)}\eps \underset{x_1,\ldots,x_{m}}{\E}\left[\underset{{\tau_4 \sim X_{T^m}(s^m_4)}}{\E}[g_i^2|_{I_4^m}(\tau_4)] \right]^{1/2}\\
    &= 2^{O(i)}\eps \underset{\tau \sim X(|I_4^m|) }{\E}\left[\langle g_i|_{\tau},g_i|_{\tau} \rangle \right]^{1/2}\\
    &=2^{O(i)}\eps\norm{g_i}_2
\end{align*}
where the last two equalities follow from Garland's method. Turning our attention to the lefthand expectation, we can apply Cauchy-Schwarz to get:
\begin{align*}
    &\underset{x_1,\ldots,x_{m}}{\E}\left[\Gamma g_i^2|_{I_2^{m-1}}(x_m)\underset{{\tau_1 \sim X_{T^{m-1}}(s^{m-1}_1)}}{\E}\left[g_i^2|_{I_1^{m-1}}(\tau_1)\right]\right]^{1/2}\\
    \leq & \ \underset{x_1,\ldots,x_{m-1}}{\E}\left[\langle \Gamma g_i^2|_{I_2^{m-1}}, \Gamma g_i^2|_{I_2^{m-1}} \rangle\right]^{1/4}\underset{x_1,\ldots,x_{m-1}}{\E}\left[\underset{{\tau_1 \sim X_{T^{m-1}}(s^{m-1}_1)}}{\E}\left[g_i^2|_{I_1^{m-1}}(\tau_1)\right]^2\right]^{1/4}\\
    \leq & \ 2^{O(i)}\gamma^{1/2}\underset{x_1,\ldots,x_{m-1}}{\E}\left[\langle g_i^2|_{I_2^{m-1}}, g_i^2|_{I_2^{m-1}} \rangle\right]^{1/4}\underset{x_1,\ldots,x_{m-1}}{\E}\left[\langle g_i^2|_{I_1^{m-1}}, g_i^2|_{I_1^{m-1}} \rangle \right]^{1/4}
    \end{align*}
    where in the last step we have applied the fact that $\norm{\Gamma} \leq O(i\gamma)$. Analysis of the remaining expectations follows exactly as before. In particular, re-arranging variables by symmetry and applying Garland's method, we can continue the above inequality as follows:
    \begin{align*}
    = & \ \gamma^{1/2} 2^{O(i)}\underset{\tau \sim X(|I_2^{m-1}|)}{\E}\left[\langle g_i^2|_{\tau}, g_i^2|_{\tau} \rangle\right]^{1/4}\underset{\tau \sim X(|I_1^{m-1}|)}{\E}\left[\langle g_i^2|_{\tau}, g_i^2|_{\tau} \rangle \right]^{1/4}\\
    = & \ \gamma^{1/2} 2^{O(i)} \langle g_i^2, g_i^2 \rangle^{1/2}\\   
    \leq & \ \gamma^{1/2} 2^{O(i)} \eps \norm{g_i}_2
\end{align*}
where in the final step we have again applied our bound on $\norm{g_i}_{\infty}$. Putting the analysis of these two terms together, we get the desired bound on the first summand of $err(g)$:
\begin{align*}
 &\underset{x_1,\ldots,x_{m}}{\E}\left[\sqrt{\Gamma g_i^2|_{I_2^{m-1}}(x_m)}\sqrt{\underset{{\tau_1 \sim X_{T^{m-1}}(s^{m-1}_1)}}{\E}\left[g_i^2|_{I_1^{m-1}}(\tau_1)\right]}\sqrt{\underset{{\tau_3 \sim X_{T^m}(s^m_3)}}{\E}[g_i^2|_{I_3^m}(\tau_3)]}\sqrt{\underset{{\tau_4 \sim X_{T^m}(s^m_4)}}{\E}[g_i^2|_{I_4^m}(\tau_4)]} \right]\\ 
    \leq& \gamma^{1/2}2^{O(i)}\eps^2 \norm{g_i}_2^2
\end{align*}
Thankfully, the analysis of second summand in $err(g)$ is exactly the same, and the third term differs only in that the lefthand expectation in the previous analysis becomes:
\[
\underset{x_1,\ldots,x_{m}}{\E}\left[\Gamma g_i^2|_{I_1^{m-1}}(x_m)\Gamma g_i^2|_{I_2^{m-1}}(x_m)\right]^{1/2} \leq \gamma 2^{O(i)}\eps\norm{g_i}_2
\]
by the same arguments. Combining these together, we get that our error term is bounded by $\gamma^{1/2} 2^{O(i)}\eps^2\norm{g_i}_2^2$, which completes the proof.
\end{proof}
It is left to analyze $H_3$ and $H_4$. Recalling that we've assumed $\{x_1,\ldots,x_j\}=H_3 \cup H_4$, \Cref{claim:H_2} can be restated as:
\[
\underset{x_1,\ldots,x_\ell}{\E}\left[g_i(I_1)g_i(I_2)g_i(I_3)g_i(I_4) \right] \leq
\underset{H_3 \cup H_4}{\E}\left[\sqrt{\prod\limits_{r=1}^4\underset{{\tau_r \sim X_{H_3 \cup H_4}(s_r^j)}}{\E}[g_i^2|_{I_r^j}(\tau_r)]} \right] + 2^{O(i)}\gamma^{1/2}\eps^2\norm{g_i}_2^2.
\]
The key is now to apply \Cref{lemma:g-link-2nd-moment}, which says that the maximum of the inner restricted expectations are small, where the factor is better the fewer variables we restrict. In order to minimize the number of restrictions, we use Cauchy-Schwarz to separate out $I_1$ and $I_2$ from $I_3$ and $I_4$:
\begin{align*}
    \underset{H_3 \cup H_4}{\E}\left[\sqrt{\prod\limits_{r=1}^4\underset{{\tau_r \sim X_{H_3 \cup H_4}(s_r^j)}}{\E}[g_i^2|_{I_r^j}(\tau_r)]} \right] &\leq \underset{H_3 \cup H_4}{\E}\left[\prod\limits_{r=1}^2\underset{{\tau_r \sim X_{H_3 \cup H_4}(s_r^j)}}{\E}[g_i^2|_{I_r^j}(\tau_r)] \right]^{1/2}\\
    &\cdot \underset{H_3 \cup H_4}{\E}\left[\prod\limits_{r=3}^4\underset{{\tau_r \sim X_{H_3 \cup H_4}(s_r^j)}}{\E}[g_i^2|_{I_r^j}(\tau_r)] \right]^{1/2}.\\
\end{align*}
Analysis of these two terms is the same, so we focus on the former. The idea is to bound one of the two inner expectations (say $I_1$) by its maximum, and note that the other term then simply returns $\norm{g_i}$. Unfortunately, there is a slight issue with this strategy naively: $H_3$ may contain variables that are not in $I_1$, so we cannot directly apply \Cref{lemma:g-link-2nd-moment}. Thankfully, localization again comes to our rescue: we can apply  \Cref{lemma:g-link-2nd-moment} as long as we first de-correlate $I_1$ from the extraneous variables in $H_3$ as in \Cref{claim:H_2}. More formally, letting 
$B_{12}=(H_3 \cap I_1 \cap I_2) \cup H_4$
for simplicity of notation, by exactly the same inductive argument used in \Cref{claim:H_2} we have:
\begin{align*}
\underset{H_3 \cup H_4}{\E}\left[\prod\limits_{r=1}^2\underset{{\tau_r \sim X_{H_3 \cup H_4}(s_r^j)}}{\E}[g_i^2|_{I_r^j}(\tau_r)] \right]^{1/2} &\leq \underset{B_{12} \sim X}{\E}\left[\underset{\tau_1 \sim X_{B_{12}}}{\E}[g_i^2|_{I_1^{B_{12}}}(\tau_1)]\underset{{\tau_2 \sim X_{B_{12}}}}{\E}[g_i^2|_{I_2^{B_{12}}}(\tau_2)] \right]^{1/2} + 2^{O(i)}\varepsilon^2\gamma^{1/2}\norm{g_i}_2
\end{align*}
where for the moment we have omitted the sizes of $B_{12}$, $\tau_1$, and $\tau_2$ for simplicity (these will be computed soon). Pulling out the maximal $I_1$ term and applying \Cref{lemma:g-link-2nd-moment} with $j=|B_{12}|$, we then get:
\begin{align*}
\underset{H_3 \cup H_4}{\E}\left[\prod\limits_{r=1}^2\underset{{\tau_r \sim X_{H_3 \cup H_4}(s_r^j)}}{\E}[g_i^2|_{I_r^j}(\tau_r)] \right]^{1/2} &\leq \max_{B_{12}}\left(\underset{\tau_1 \sim X_{B_{12}}}{\E}[g_i^2|_{I_1^{B_{12}}}(\tau_1)]\right)^{1/2}\norm{g_i}_2 + 2^{O(i)}\varepsilon^2\gamma^{1/2}\norm{g_i}_2\\
&\leq \frac{2^{O(i)}\eps^{1/2}}{{k-|B_{12}| \choose i-|B_{12}|}^{1/2}}\norm{g_i}_2 + 2^{O(i)}\varepsilon^2\gamma^{1/2}\norm{g_i}_2
\end{align*}
The same argument holds for the latter product over $I_3$ and $I_4$. Letting 
$B_{12}=(H_3 \cap I_3 \cap I_4) \cup H_4$, and putting everything together, we finally get the bound:
\begin{align*}
\underset{x_1,\ldots,x_\ell}{\E}\left[g_i(I_1)g_i(I_2)g_i(I_3)g_i(I_4) \right] &\leq \frac{2^{O(i)}\eps}{{k-|B_{12}| \choose i-|B_{12}|}^{1/2}{k-|B_{34}| \choose i-|B_{34}|}^{1/2}}\norm{g_i}_2^2 + 2^{O(i)}\varepsilon^2\gamma^{1/2}\norm{g_i}_2^2\\
&\leq \frac{2^{O(i)}\eps}{{k-|B_{12}| \choose i-|B_{12}|}^{1/2}{k-|B_{34}| \choose i-|B_{34}|}^{1/2}}\frac{1}{{k \choose i}}\norm{f_i}_2^2 + 2^{O(i)}\varepsilon^2\gamma^{1/2}\norm{D^k_i f}_2^2\\
&\leq 2^{O(i)}\eps \left(\frac{k}{i}\right)^{\frac{|B_{12}|+|B_{34}|}{2} -2i}\norm{f_i}_2^2 + 2^{O(i)}\varepsilon^2\gamma^{1/2}\norm{D^k_i f}_2^2
\end{align*}
where we have applied the basic binomial bound ${n \choose p} \geq \left(\frac{n}{p}\right)^p$. To complete the result, it suffices show that $\frac{|B_{12}|+|B_{34}|}{2} = 2i-\ell$. This follows similarly to the analogous argument in \cite{khot2018small}, but we'll give a simplification of their proof for completeness. Recall that $B_{12}$ consists of variables in $H_3$ and $H_4$ that appear in both $I_1$ and $I_2$, and $B_{34}$ similarly consists of variables in $H_3$ and $H_4$ that appear in both $I_3$ and $I_4$. Since every variable in $H_3$ occurs in exactly one of $(I_1 \cap I_2)$ and $(I_3 \cap I_4)$ by definition, we get that 
\[
\frac{|B_{12}|+|B_{34}|}{2} = |H_4| + \frac{|H_3|}{2}.
\]
To compute the righthand side, note that by definition we have the following two relations:
\begin{enumerate}
    \item Since each term has $4i$ total variables (with repetition): 
    \[
        4|H_4| + 3|H_3| + 2|H_2| = 4i
    \]
    \item Since there are $\ell$ unique variables:
    \[
        |H_4| + |H_3| + |H_2| = \ell
    \]
\end{enumerate}
Combining these equations gives the desired equality:
\[
|H_4| + \frac{|H_3|}{2} = \frac{(4|H_4| + 3|H_3| + 2|H_2|) - 2(|H_4| + |H_3| + |H_2|)}{2} = 2i - \ell.
\]
Putting everything together, we finally get
\[
\underset{x_1,\ldots,x_\ell}{\E}\left[g_i(I_1)g_i(I_2)g_i(I_3)g_i(I_4) \right] \leq 2^{O(i)}\eps\left(\frac{i}{k}\right)^{\ell}\norm{f_i}^2 + c_{k,i}\varepsilon\gamma^{1/2}\norm{D^k_i f}^2,
\]
as desired.
\section{Characterizing Expansion in HD-walks}\label{sec:expansion}
One traditional application of hypercontractivity on the discrete hypercube lies in showing that the noisy hypercube graph (given by randomizing each bit of a binary string $x$ with some probability $1-\rho$) is a small-set expander. This result is also often thought of as stating ``sparse functions on the hypercube are noise-sensitive,'' an interpretation we'll discuss in the next section. Unlike the noisy hypercube, it is well known that HD-walks are far from being small set expanders \cite{bafna2020high}. Before we quantify this further, let's recall the definition of (edge) expansion in the general weighted setting.
\begin{definition}[Weighted Edge Expansion]
Let $(X,\Pi)$ be a weighted simplicial complex, $M$ a $k$-dimensional HD-Walk over $(X,\Pi)$, and $S \subset X(k)$ a subset. The weighted edge expansion of $S$ is the average probability of leaving $S$ after one step of $M$:
\[
\Phi(S) = \underset{v \sim \pi_k|_S}{\mathbb{E}}\left[ M(v,X(k) \setminus S)\right],
\]
where $\pi_k|_S$ is the (normalized) restriction of $\pi_k$ to $S$,
\[
M(v, X(k) \setminus S) = \sum\limits_{y \in X(k) \setminus S} M(v,y),
\]
and $M(v,y)$ is the transition probability from $v$ to $y$.
\end{definition}
A small-set expander is simply a graph where all small sets expand. To understand why HD-walks fail this condition, let's consider the Johnson graph. The Johnson graph $J(n,k,\ell)$ is the graph on ${[n] \choose k}$ whose edges are given by sets with intersection size $\ell$. Well-studied object in their own right, the Johnson graphs are a fundamental example higher order random walks on the complete complex \cite{alev2019approximating}. In our context, we usually think of $n$ as being much larger than $k$, and $\ell$ as being (at least) $ck$ for some constant $0<c<1$. In this case, one can show by direct computation that the expansion of any $i$-link $X_\tau$ is bounded away from $1$:
\[
\Phi(X_\tau) \approx 1 - c^{-i},
\]
despite the fact that its density is vanishingly small: $\mathbb{E}[1_{X_\tau}] \approx (k/n)^{i}$.

Recently, BHKL proved a general variant of this result for all HD-walks (see \cite[Theorem 9.2]{bafna2020high}). They show that spectrum of any $k$-dimensional walk $M$ on a sufficiently strong local-spectral expander is divided up into $k+1$ strips of width $O_{k,M}(\gamma)$\footnote{BHKL actually only prove the width is $O_{k,M}(\sqrt{\gamma})$, the improvement to $O_{k,M}(\gamma)$ was given soon after by Zhang \cite{Zhang2020}.} centered around some set of approximate eigenvalues $\{\lambda_i(M)\}_{i=0}^k$, and that the expansion of any link at level $i$ is almost exactly $1-\lambda_i(M)$. They also prove a weak converse to this result: \textit{any} non-expanding set must be concentrated in a link. It is convenient to state the contrapositive. For any $\delta>0$, let $R_\delta(M)=r$ denote the number of approximate eigenvalues of $M$ that are greater than $\delta$ (a quantity BHKL call the \textbf{ST-Rank} of $M$). BHKL \cite[Theorem 9.5]{bafna2020high} prove that the expansion of any set $S \subset X(k)$ is at least:
\begin{equation}\label{eq:BHKL}
\Phi(S) \gtrsim 1 - \delta - c_1{k \choose r}\varepsilon - c_2\gamma
\end{equation}
where $S$ is $(\varepsilon,r)$-pseudorandom.\footnote{Note that we have simplified BHKL's result here somewhat for simplicity of presentation, but it is an accurate representation of their result in most cases of interest.} This is great when $\varepsilon \ll {k \choose r}$, but for many applications of interest (e.g.\ in hardness of approximation), we think of $\varepsilon$ as fixed and of $k$ as going to infinity. In this regime, the above characterization is useless, as the bound reduces to the trivial fact $\Phi(S) \geq 0$. Using hypercontractivity, we can completely resolve this issue by offering a variant of \Cref{eq:BHKL} with no dependence on $k$. Before we give the statement, however, we note that both BHKL and our result require the approximate eigenvalues of the HD-walk $\{\lambda_i(M)\}_{i=0}^k$ to decrease monotonically. BHKL proved that this property holds for a broad class of walks they call \textbf{complete walks}, which includes all HD-walks of interest studied in the literature.
\begin{definition}[Complete HD-Walk (\cite{bafna2020high} Definition 7.10)]\label{def:complete-walk}
Let $(X,\Pi)$ be a weighted, pure simplicial complex and $M=\sum\limits_{Y \in \mathcal Y} \alpha_Y Y$ an HD-walk on $(X,\Pi)$. $M$ is called \textit{complete} if for all $n \in \mathbb{N}$ there exist $n_0 > n$ and $d$ such that $\sum\limits_{Y \in \mathcal Y} \alpha_Y Y$ is also an HD-walk when taken to be over the $d$-dimensional complete complex on $n_0$ vertices.
\end{definition}
All walks we have seen so far (canonical walks, swap walks, pure walks, affine combinations thereof, etc.) are complete, so restricting to this class does not lose much generality. With this in mind, we can finally state our dimension independent bound on the expansion of pseudorandom sets.

\begin{theorem}[Pseudorandom Sets Expand]\label{thm:expansion}
Let $(X,\Pi)$ be a two-sided $\gamma$-local-spectral expander, $M$ a complete $k$-dimensional HD-walk, and $S \subseteq X(k)$ of density $\alpha$. Then for any $\delta>0$ and $r=R_\delta(M)-1$, the expansion of $S$ is at least:
\[
\Phi(S) \geq 1 - \delta - (1-\delta)2^{O(r)} \eps^{1/3} - c\gamma
\]
where $c \leq 2^{O(k)}w(M)h(M)^2$ and $S$ is $(\varepsilon,r)$-pseudorandom.
\end{theorem}

The proof of \Cref{thm:expansion} goes through a level-$i$ inequality for pseudorandom functions of independent interest.
\begin{theorem}[Level-$i$ Inequality]\label{thm:level-i}
Let $(X,\Pi)$ be a $\gamma$-local-spectral expander with $\gamma < 2^{-\Omega(k)}$ and $f \in C_k$ a boolean,  $(\eps,i)$-pseudorandom function. Then:
\[
\langle f, \fup{i} \rangle \leq 2^{O(i)}\eps^{1/3}\mathbb{E}[f].
\]
\end{theorem}
Let's first prove \Cref{thm:expansion} given \Cref{thm:level-i}.
\begin{proof}[Proof of \Cref{thm:expansion}]
The argument is standard and follows from the  identity $\Phi(S) = 1- \frac{1}{\alpha}\langle f, Mf \rangle$ (where $\alpha=\mathbb{E}[f]$),
and expanding $f = \sum\limits_{i=0}^k \fdown{i}$ (the HD-Level-Set Decomposition). Namely, since $\fdown{i}$ is an approximate eigenvector, we can write:
\begin{align*}
    \Phi(S) &= 1- \frac{1}{\alpha}\sum\limits_{i=0}^k \langle f, M\fdown{i} \rangle\\
    &\geq 1- \frac{1}{\alpha}\sum\limits_{i=0}^k \lambda_i(M)\langle f, \fdown{i} \rangle - c_1\gamma\\
\end{align*}
where $c_1 \leq w(M) h(M)^2 2^{O(k)}$.
We can now apply \Cref{thm:Bottom-vs-Top} to switch between decompositions to get:
\[
\Phi(S) \geq 1- \frac{1}{\alpha}\sum\limits_{i=0}^k \lambda_i(M)\langle f, \fup{i} \rangle - c_2\gamma
\]
where $c_2 = c_1+2^{O(k)}$. Since $M$ is a complete walk, its eigenvalues decay monotonically \cite{bafna2020high}, we can therefore simplify the above to:
\begin{align*}
    \Phi(S) &\geq 1 - \frac{1}{\alpha}\sum\limits_{i=0}^{r} \lambda_i(M)\langle f, \fup{i} \rangle - \frac{\lambda_{r+1}(M)}{\alpha}\sum\limits_{i=r+1}^k \langle f, \fup{i} \rangle - c_2\gamma\\
    &= 1 -
    \frac{1}{\alpha}\sum\limits_{i=0}^{r} \lambda_i(M)\langle f, \fup{i} \rangle - \frac{\lambda_{r+1}(M)}{\alpha}\left (\alpha -  \sum\limits_{i=0}^{r}\langle f, \fup{i} \rangle\right) - c_2\gamma.
\end{align*}
Recall that by definition $\lambda_{r+1}(M) \le \delta$, and hence
\begin{align*}
    \Phi(S) &\geq 1 - \delta - \frac{1-\delta}{\alpha}\sum\limits_{i=0}^r \langle f,\fup{i} \rangle - c_2\gamma\\
    &\geq 1 - \delta - (1-\delta)2^{O(r)}\varepsilon^{1/3} - c_2\gamma
\end{align*}
where in the last step we have applied \Cref{thm:level-i}.
\end{proof}

It is left to prove \Cref{thm:level-i}, which also follows from fairly standard arguments given \Cref{thm:hypercontraction}. 
\begin{proof}[Proof of \Cref{thm:level-i}]
To simplify notations, we write $f_i$ instead of $\fup{i}$.
Notice that by H\"{o}lder's inequality for $p=4/3, q=1/4$ we have:
\begin{align*}
\ip{f,f_i} \le \|f\|_{4/3} \|f_i\|_{1/4} = \alpha^{3/4} \E[f_i^4]^{1/4}
\end{align*}
Combining this with \Cref{thm:hypercontraction} gives the following relation:
\begin{equation}\label{eq:level}
\frac{\langle f, f_i \rangle^4}{\alpha^3} \leq \E[f_i^4] \leq 2^{O(i)}\eps\E[f_i^2] + c_1\gamma^{1/2}\alpha \leq 2^{O(i)}\eps\langle f,f_i \rangle + c_2\gamma^{1/2}\alpha
\end{equation}
where $c_1,c_2 \leq 2^{O(k)}$
by approximate orthogonality (\Cref{lemma:approx-orthog}). We can simplify the above via two observations. First, note that we can assume without loss of generality that $\gamma^{1/4} \leq c_1^{-1} \eps$. This follows from observing that:
\[
\langle f, f_i \rangle = {k \choose i}\sum\limits_{j=0}^i (-1)^{i-j}{i \choose j}\langle D^k_j f, D^k_j f \rangle.
\]
Appealing to arguments from \cite[Lemma 8.8]{bafna2020high}, we have that $\langle D^k_j f, D^k_j f \rangle \leq \varepsilon\alpha$, which gives the naive bound:
\[
\langle f, f_i \rangle \leq {k \choose i}2^{O(i)}\varepsilon\alpha.
\]
If $\varepsilon \le {k \choose i}^{-3/2}$
then ${k \choose i} \eps \le \eps^{1/3}$ and our desired bound follows. Otherwise, we may assume from now on that $\varepsilon \ge {k \choose i}^{-3/2} \ge 2^{-(3/2)k}$.
Since $\gamma \leq 2^{-\Omega(k)}$, we are therefore free to assume $\gamma^{1/4} \le c_1^{-1} \eps$ as well. Second, we can also assume $\langle f,f_i \rangle \geq \gamma^{1/4}\alpha$, since otherwise we are done by our previous assumptions on $\gamma$ and $\varepsilon$. Combining these with \Cref{eq:level} then gives:
\begin{align*}
\frac{\langle f, f_i \rangle^4}{\alpha^3} &\leq 2^{O(i)}\eps\langle f,f_i \rangle + c_1\gamma^{1/2}\alpha\\
&\leq 2^{O(i)}\eps\langle f,f_i \rangle + \varepsilon\gamma^{1/4}\alpha\\
&\leq 2^{O(i)}\varepsilon\langle f, f_i \rangle \end{align*}
as desired.
\end{proof}

\section{Fourier Analysis on HDX}\label{sec:Fourier}
In this section we further develop the theory of Fourier analysis on simplicial complexes, and show how hypercontractivity for pseudorandom functions (\Cref{thm:hypercontraction}) recovers tight analogs of the KKL Theorem and noise-sensitivity of sparse functions. This requires introducing a number of new analog definitions of classic Fourier analytic quantities on simplicial complexes. To get an idea for what these should look like, it will be useful to start by considering a natural embedding of the hypercube itself into a simplicial complex. 
\begin{definition}[Hypercube Complex]
The hypercube complex $X=X_{\{0,1\}^n}$ is the complete $n$-partite complex on $X(1) = [n] \times \{0,1\}$, where the first coordinate denotes the color of the vertex. That is, the top level faces are $X(n)=\{\{(1,x_1),\ldots,(n,x_n)\}: x \in \{0,1\}^n\}$.
\end{definition}
We make a few notes on this definition. First, it is clear from definition that $X(n)$ can equivalently be thought of as the hypercube $\{0,1\}^n$, where each color in $[n]$ corresponds to a coordinate in $\{0,1\}^n$. Further, classic graphs on $\{0,1\}^n$ such as the hypercube or noisy hypercube can be expressed as simple higher order random walks. The hypercube graph, for instance, is simply the non-lazy lower walk $UD^+ = 2U_{n-1}D_n - I$. This embedding will serve as our guiding principle for developing analog Fourier-analytic definitions on simplicial complexes---whenever possible, our definitions will reduce to the standard notion when applied to the hypercube complex. We note that the same embedding can be used for any product distribution and all of our definitions will generalize appropriately. We focus on the simple case of the hypercube for ease of exposition.

\subsection{Total Influence and the KKL Theorem}
We'll start with a fundamental concept in classical Fourier analysis, \textit{influence}. Let's first recall the definition of (total) influence on the discrete hypercube. Influence can be formalized in a number of equivalent ways. It is often thought of, for instance, as a measure of average sensitivity. In our context, it will be most convenient to view influence as a statement about the expansion of a function with respect to the hypercube graph. More formally, let $Q_n$ denote the normalized adjacency matrix of the hypercube graph, and $Q^{\lazy}_n = \frac{I+Q_n}{2}$ its lazy variant. We will write total influence in terms of the (un-normalized) \textit{Laplacian operator} $L=n(I-Q^{\lazy}_n)$.
\begin{definition}[Total Influence (hypercube)]
Let $f: \{0,1\}^n \to \{0,1\}$ be a Boolean function. The total influence of $f$, denoted $I[f]$, is:
\begin{align*}
    I[f] = \langle f, L f \rangle.
\end{align*}
\end{definition}
Expressed in this sense, there is a natural generalization to simplicial complexes. It is not hard to see that on the hypercube complex, $Q_n^{\lazy}$ is exactly the lower walk $UD$. As a result, we'll define influence using the Laplacian of the lower walk.
\begin{definition}[Total Influence]\label{def:influence}
Let $(X,\Pi)$ be a pure, weighed simplicial complex and $f \in C_k$. The influence of $f$, denoted $I[f]$ is:
\[
I_{(X,\Pi)}[f] = \langle f,L_{UD}f \rangle
\]
where $L_{UD}=k(I-U_{k-1} D_k)$.
When clear from context, we will simply write $I[f]$.
\end{definition}

When $(X,\Pi)$ is sufficiently expanding, \Cref{def:influence} acts much like standard influence on the cube. For instance, recalling standard bounds on the spectral expansion of $L_{UD}$ \cite{kaufman2020high}, it is not hard to see the total influence of any function on a $\gamma$-local-spectral expander lies between $(1+O_k(\gamma))\var(f) \leq I_{(X,\Pi)}[f] \leq k\var(f)$, which returns the standard bounds as $\gamma$ goes to $0$. Similarly, it is obvious that the total influence of any function on the hypercube complex is equivalent to its total influence on the hypercube, as the lower walk $U_{n-1} D_n$ is exactly $Q_n^{\lazy}$.
\begin{observation}
Let $f: \{0,1\}^n \to \mathbb{R}$ be any function and $f_X: X_{\{0,1\}^n}(n) \to \mathbb{R}$ its equivalent on the hypercube complex, that is:
\[
f(x_1,\ldots,x_n) = f_X( (1,x_1), \ldots, (n,x_n) )
\]
Then:
\[
I_{X_{\{0,1\}^n}}[f] = I[f].
\]
\end{observation}

One of the most well-studied problems in the analysis of boolean functions is understanding the structure of functions with low influence. The seminal result in this area is called the ``KKL Theorem'' \cite{kahn1988influence}. Informally, the KKL Theorem states that if a function has low total influence, there must exist an \textit{influential coordinate} (in the sense that on average over $\{0,1\}^n$, the coordinate has a large affect on the value of $f$). Morally, this can also be thought of as strong notion of the following statement: ``functions with low influence are not pseudorandom.'' While the KKL Theorem itself does not extend beyond the hypercube, this latter interpretation does. In particular, Bourgain \cite{friedgut1999sharp} proved a similar statement over any product space: functions with low influence must have some influential \textit{set} of coordinates, and are therefore not pseudorandom. We prove a variant of Bourgain's result for local-spectral expanders.

\begin{theorem}[Bourgain's Theorem for HDX]\label{thm:hdx-sharp-threshold}
Let $(X,\Pi)$ be a two-sided $\gamma$-local-spectral expander with $\gamma \leq 2^{-\Omega(k)}$ and $f \in C_k$ a boolean function. Then for any $0 \leq K \leq k$, if $I[f] \leq K\var(f)$, there exists an $(i \leq K)$-link $\tau$ with large density:
\[
 \underset{X_\tau}{\E}[f] \geq 2^{-O(K)}.
\]
\end{theorem}

\begin{proof}
This follows without too much difficulty from the expansion of pseudorandom sets (\Cref{thm:expansion}). In particular, notice that our assumption on the influence implies the following bound on the expansion of $f$ with respect to the lower walk $U_{k-1}D_k$:
\[
\Phi(f) = \frac{\ip{f,L_{UD}f}}{k\E[f]} = \frac{I[f](1 - \E[f])}{k\var(f)} \leq \frac{K}{k}.
\]
Recall that \Cref{thm:expansion} states that for any $\delta>0$ and $r = R_\delta(UD)-1$, the expansion of an $(\varepsilon,r)$-pseudorandom boolean function $g$ with respect to the lower walk is at least:
\[
\Phi(g) \geq 1 - \delta - (1-\delta)2^{O(r)}\varepsilon^{1/3} - c\gamma.
\]
Using this fact, we'll show that $f$ cannot be $(\varepsilon,K)$-pseudorandom for $\varepsilon \leq 2^{-\Omega(K)}$, which gives the result.

To this end, assume $f$ is $(\varepsilon, K)$-pseudorandom for some $\varepsilon=2^{-\Omega(K)}$ to be determined soon (else we are done), and let $\delta$ be $1 - \frac{K(1 + \eps^{1/6})}{k}$ such that $1 - \delta < (K+1)/k$. Since the eigenvalues of $UD$ are concentrated around $1, 1 - 1/k, 1 - 2/k, \cdots, 1 - K/k$ for small enough $\gamma$ \cite{bafna2020high}, the ST-Rank $R_\delta(UD)=K+1$, and $r=K$. \Cref{thm:expansion} then implies:
\[
\Phi(f) \geq \frac{K}{k} \cdot (1+\eps^{1/6})(1 - 2^{O(r)}\eps^{1/3}),
\]
where we have again used our assumption on the size of $\gamma$. Re-arranging the above using the upper bound on expansion then gives a lower bound on $\varepsilon$ of:
\[
\eps^{1/3} \geq \frac{\eps^{1/6}}{1+\eps^{1/6}}\cdot \frac{1}{2^{O(r)}} \geq \frac{1}{2^{O(r)}},
\]
which implies the result.



\end{proof}

Before moving on, we'll prove that this result is tight.
\begin{proposition}\label{thm:lower}
Let $c \geq 1$ be any constant. Then for all integers $K,k >1$ satisfying $k \geq \Omega_c(K)$ 
and any $n$ sufficiently larger than $k$, there exists a Boolean function $f \in C_k$ on the $k$-dimensional complete complex on $n$ vertices satisfying:
\begin{enumerate}
    \item The influence of $f$ is small:
    \[
    I[f] \leq K\text{Var}(f)
    \]
    \item For every $i \leq cK$, all $i$-links are sparse:
    \[
\forall i \leq cK, \tau \in X(i): \underset{X_\tau}{\mathbb{E}}[f] \leq 2^{-\Omega(K)}.
\]
\end{enumerate}
\end{proposition}
\begin{proof}
Our construction is based on a careful analysis of the anti-tribes function (a.k.a the AND of ORs function) similar to \cite[Example 5.8]{keevash2019hypercontractivity}. Concretely, let $T_1,\ldots T_{m}$ (called ``tribes'') be $m=2cK$ disjoint sets of $c_1\frac{n}{k}$ vertices for some $c_1 \geq \log(\Omega(c))$.
We define our candidate function $f \in C_k$ to be $1$ on a $k$-face $S$ exactly when $S$ contains some vertex from each $T_i$:
\[
f(S)=
\begin{cases}
1 & \text{if $\forall 1 \leq i \leq m$: $|S \cap T_i| > 0$}\\
0 & \text{else}.
\end{cases}
\]
For simplicity, it will actually be more convenient to analyze $f$ as a function over $[n]^k$ rather than $X(k)={n \choose[k]}$. Since the probability of repeated vertices in the former is $o_{n,k}(1)$, this has no effect on our final result when $n$ is sufficiently larger than $k$.

Let's start by proving the density of every $cK$-link is at most $2^{-\Omega(k)}$.\footnote{Formally we note this should really be shown for $\floor{cK}$-links, but this makes no significant difference in the analysis so we ignore it for simplicity.} Note that this implies the same for every $i$-link for $i \leq cK$. It is not hard to see that the largest density link comes from fixing an element in each of $cK$ tribes. For simplicity, fix such a $cK$-link $T$ with a vertex in each $T_i$ for $cK+1 \leq i \leq m$ (all such links are symmetric, so it suffices to analyze this case). For a uniformly drawn element $S \in [n]^k$, let $E_i$ denote the event that $S$ contains a vertex in $T_i$. We'd like to bound:
\[
\underset{X_T}{\mathbb{E}}[f]= \Pr_{S \sim [n]^k}\left[\bigcap_{i=1}^{cK} E_i ~\Bigg|~ S \supset T \right] = \Pr_{S \sim [n]^{k-cK}}\left[\bigcap_{i=1}^{cK} E_i \right],
\]
where we have used the fact that $S \setminus T$ is independent of $T$ since we are working over $[n]^k$.
Since the $E_i$ are negatively correlated, we can bound this probability by:
\begin{align*}
    \Pr_{S \sim [n]^{k-cK}}\left[\bigcap_{i=1}^{cK} E_i \right]  & \leq \prod_{i=1}^{cK} \Pr_{S \sim [n]^{k-cK}}\left[ E_i \right]\\
    &\leq (1-(1-c_1/k)^{k})^{cK}\\
    & \leq \left(1 -\frac{1}{O(c)}\right)^{cK}\\
    & \leq 2^{-\Omega(K)}
\end{align*}
where we've used the fact that $e^{-x} \geq 1-x \geq e^{-x/(1-x)}$ for $x<1$ and our assumptions on the size of $k$.


We now move on to analyzing the influence of $f$, which will follow from similar computations. To start, notice that it is instead sufficient to bound the expansion of $f$ with respect to the lower walk by: 
\[
\Phi(f) \leq \frac{K}{k}\frac{\text{Var}(f)}{\mathbb{E}[f]} = \frac{K(1-\mathbb{E}[f])}{k},
\] 
as then:
\begin{align*}
I[f] &= \langle f,L_{UD}f\rangle = k\Phi(f)\mathbb{E}[f] \leq K\text{Var}(f),
\end{align*}
where we recall $L_{UD}$ is the un-normalized Laplacian of $UD$. 

To this end, recall that the expansion of $f$ can also be defined as the average probability of leaving $\text{supp}(f)$ after applying the walk, that is: 
\[
\Phi(f) = \underset{S \sim \text{supp}(f)}{\mathbb{E}}[\phi(S)],
\]
where $\phi(S)$ denotes the probability of leaving $S$ in a single step of the lower walk. To compute this value, recall that in the down step of the walk, a uniformly random vertex is removed from $S$. In order to leave the support of $f$ in the up step, the removed element must have been selected from a tribe $T_i$ such that $|S \cap T_i| = 1$. The idea is then to show that for most samples, only a small fraction of tribes have exactly one element. With this in mind, let $B_i$ be the event $|S \cap T_i| = 1$ over the randomness of $S \sim \text{supp}(f)$. Formalizing the above argument, we can bound $\phi(S)$ by the sum over $B_i$:
\[
\phi(S) \leq \sum_{i=1}^{m} \frac{B_i(S)}{k},
\]
and therefore the expansion $\Phi(f)$ by:
\[
\Phi(f) \leq \frac{1}{k}\underset{{S \sim \text{supp}(f)}}{\E}[B_i(S)].
\]

By a similar argument to our density calculations, the probability that any fixed tribe $T_i$ has exactly one element from $S \sim \text{supp}(f)$ is at most:
\begin{align*}
\mathbb{E}[B_i] &= \left(1-\frac{c_1}{k}\right)^{k-m}\\ 
&\leq e^{-c_1\frac{k-m}{k}}\\ 
&\leq \frac{1}{\Omega(c)}
\end{align*}
since we have by assumption that $k$ is much larger than $m$. Plugging this into our expression for expansion then gives:
\[
\Phi(f) \leq \frac{m}{k} \cdot \frac{1}{\Omega(c)} \leq c_2\frac{K}{k}
\]
for some $c_2 < 1$. Noting that $\mathbb{E}[f] = 2^{-\Omega(K)}$ then implies the result for the appropriate setting of constants.
\end{proof}

\subsection{Stability and the Noise Operator}
Another fundamental notion in boolean Fourier analysis is the \textit{noise operator} $T_\rho$. It is convenient to express the definition in terms of the following process on an element $x \in \{0,1\}^n$:
\begin{enumerate}
    \item Remove each bit with probability $1-\rho$.
    \item Replace each removed bit uniformly\footnote{In more general settings like the p-biased cube, this is replaced with respect to the underlying distribution.} at random.
\end{enumerate}
We write the distribution over $y$ given by this process as $N_\rho(x)$. The noise operator $T_\rho$ is simply the averaging operator over $\rho$-correlated strings.
\begin{definition}[Noise Operator (Hypercube)]
Let $f: \{0,1\} \to \mathbb{R}$ be any function. The noise operator $T_\rho$ averages $f$ over $N_\rho$:
\[
T_\rho f(x) = \underset{y \sim N_\rho(x)}{\E}[f(y)].
\]
\end{definition}
Extending the noise operator to simplicial complexes is a bit tricky naively since there is no notion of coordinates. To do this, consider the following reformulation of the distribution $N_\rho(x)$, instead of removing each coordinate independently with probability $1-\rho$, we remove a uniformly random \textit{set of $i$ coordinates} with probability ${n \choose i}\rho^{n-i}(1-\rho)^i$, and replace them uniformly at random. This equivalent process does have a natural analog on simplicial complexes: simply replace ``uniformly random set of $i$ coordinates'' with ``uniformly random $i$-face.'' We can formalize this through the averaging operators.
\begin{definition}[Noise Operator (Simplicial Complex)]
Let $(X,\Pi)$ be a pure, weighted simplicial complex. The noise operator $T^k_\rho(X,\Pi)$ at level $k$ of the complex is:
\[
T^k_\rho(X,\Pi) = \sum\limits_{i=0}^k {k \choose i}(1-\rho)^i \rho^{k-i}U_{k-i}^kD^k_{k-i}.
\]
We write $T_\rho$ when the level and complex are clear from context.
\end{definition}
Let's take a moment to check that, as with influence, when applied to the hypercube complex this definition recovers $T_\rho$.
\begin{observation}
Let $f: \{0,1\}^n \to \mathbb{R}$ be any function and $f_X: X_{\{0,1\}^n}(n) \to \mathbb{R}$ its equivalent on the hypercube complex, then:
\[
T^n_\rho(X_{\{0,1\}^n})f_X = T_\rho f.
\]
\end{observation}
\begin{proof}
$T^n_\rho(X_{\{0,1\}^n})$ is also an averaging operator, so it is enough to confirm it averages over the $\rho$-noisy distribution $N_\rho$. We claim this is clear from definition. In particular, notice that $U_{n-i}^nD^n_{n-i}$ on $X_{\{0,1\}^n}$ is exactly the process of removing $i$-coordinates uniformly at random, and replacing them with uniformly random bits. As we mentioned above, this is an equivalent way to define $N_\rho(x)$, is applying this process with probability ${n \choose i}(1-\rho)^i \rho^{n-i}$, which exactly matches the definition of $T_\rho(X_{\{0,1\}^n})$.
\end{proof}
The noise operator has a wide variety of applications across boolean Fourier analysis. One classical application is to analyze the \textit{noise-sensitivity} of a boolean function, that is the likelihood that the function flips on a noisy input. It is convenient to define the opposite concept first, \textit{stability}.
\begin{definition}[Stability (Hypercube)]
Let $f: \{0,1\} \to \mathbb{R}$ be any function. The stability of $f$ with respect to $\rho$, denoted $\stab_{\rho}(f)$, is:
\[
\stab_{\rho}(f) = \langle f, T_\rho f \rangle.
\]
\end{definition}
Since we already defined $T_\rho$ on simplicial complexes, stability has an obvious analog.
\begin{definition}[Stability (Simplicial Complex)]
Let $(X,\Pi)$ be a weighted, pure simplicial complex and $f\in C_k$. The noise stability of $f$ with respect to $\rho$, denoted $\stab_{\rho}(f)$, is:
\[
\stab^{(X,\Pi)}_{\rho}(f) = \langle f, T_\rho^k(X,\Pi) f \rangle.
\]
We drop $(X,\Pi)$ from the notation when clear from context.
\end{definition}
Similarly, it is clear that our definition of stability for complexes returns the original definition when applied to the hypercube complex.
\begin{observation}
Let $f: \{0,1\}^n \to \mathbb{R}$ be any function and $f_X: X_{\{0,1\}^n}(n) \to \mathbb{R}$ its equivalent on the hypercube complex, then:
\[
\stab_{\rho}(f) = \stab^{X_{\{0,1\}^n}}_{\rho}(f_X)
\]
\end{observation}
A function is called \textit{noise-sensitive} if is has poor stability. One classical result in boolean Fourier analysis is that sparse functions on the hypercube are noise-sensitive, which is equivalent to saying that the noisy hypercube graph is a small-set expander. Since the noise operator is just a specific instance of a (complete) higher order random walk, \Cref{thm:expansion} implies an analogous statement for functions on HDX: pseudorandom functions are noise-sensitive.
\begin{corollary}[Pseudorandom functions are noise sensitive]
Let $(X,\Pi)$ be a two-sided $\gamma$-local-spectral expander, $f \in C_k$ an $(r,\delta)$-pseudorandom boolean function for $r=\log(2/\eps)/\log(1/\rho)+2$ and $\delta \leq 2^{-\Omega(r)}\eps^3$. Then $f$ is noise sensitive:
\[
\stab_\rho(f) \leq (\eps + c\gamma)\E[f]
\]
for $c \leq 2^{O(k)}$.
\end{corollary}
\begin{proof}
One can directly compute from \cite[Corollary 7.6]{bafna2020high} that the approximate eigenvalues of $T_\rho$ are exactly $\lambda_i = \rho^i$. As a result, for small enough $\gamma$, the $(\varepsilon/2)$-ST-Rank of $T_\rho$ is at most: 
\[
R_{\varepsilon/2}(T_\rho) \leq \log(\varepsilon/2)/\log(1/\rho)+2.
\] 
Since $T_\rho$ is a higher order random walk, \Cref{thm:expansion} states that the non-expansion of any $(\delta,r)$-pseudorandom function $f$ of density $\alpha$ is at most:
\begin{align*}
\frac{1}{\alpha}\langle f, T_\rho f \rangle &\leq \alpha + (1-\alpha)\varepsilon/2 + 2^{O(r)}\delta + c\gamma\\ 
&\leq \alpha + (1-\alpha)\varepsilon/2 + \varepsilon/4 + c\gamma\\
&\leq \varepsilon + c\gamma,
\end{align*}
where we've used the fact that $\alpha \leq \delta \leq \varepsilon/4$.
\end{proof}

The noise operator is actually also commonly used to define hypercontractivity. In this form, the standard hypercontractive inequality generally states:
\[
\norm{T_\rho f}_4 \leq \norm{f}_2
\]
for some $\rho=\Theta(1)$. It is well known that on the hypercube this statement is in fact equivalent to Bonami's lemma. We can show a similar equivalence between our variant of Bonami's lemma (\Cref{thm:hypercontraction}) and a noise operator based form of hypercontractivity for pseudorandom functions. To state the strongest form of the result, it will be useful to extend the classic notion of \textit{degree} to simplicial complexes.
\begin{definition}[Function Degree]
Let $(X,\Pi)$ be a pure, weighted simplicial complex, and $f \in C_k$ any function. The degree of $f$, denoted $\deg(f)$, is the largest $i$ such that $\fup{i}$ is non-zero.
\end{definition}
We now show how to translate \Cref{thm:hypercontraction} into noise operator form for degree $i$, $(\varepsilon,i)$-pseudorandom functions.
\begin{proposition}\label{thm:noise-hyper}
Let $(X,\Pi)$ be a two-sided $\gamma$-local-spectral expander satisfying $\gamma \leq 2^{-\Omega(k)}$ and $f \in C_k$ a degree $i$,\footnote{When $i \ll k$, we can replace the condition on $\gamma$ with $\gamma \leq k^{-\Omega(i)}$} $(\varepsilon,i)$-pseudorandom function.
Then for some constant $\rho=\Theta(1)$, we have:
\[
\norm{T_\rho f}_4^4 \leq \eps\norm{f}_2^2\norm{f}_\infty^{2}.
\]
\end{proposition}
\begin{proof}
The overall proof follows a fairly standard reduction from hypercontractivity to Bonami's lemma (see e.g.\ \cite[Exercise 9.6]{o2014analysis}), but requires some extra work due to the fact that the Bottom-Up decomposition is only approximately an eigenbasis for $T_\rho$ (in an $\ell_2$-sense). Namely, by \cite[Proposition 7.5]{bafna2020high} and \Cref{lemma:approx-kernel}, we can write:
\[
T_\rho \fup{j} = \rho^k \fup{j} + \text{err}_j
\]
where $\norm{\text{err}_j}_2 \leq k^{O(j)}\gamma\norm{f}_2$, and $\norm{\text{err}_j}_\infty \leq k^{O(j)}\varepsilon\norm{f}_\infty$. The last of these facts is slightly less standard, and follows from noting that $\text{err}_j$ is really a linear combination of at most $k^{O(j)}$ averaging operators applied to $g_j$ (see \cite[Proposition 7.5]{bafna2020high}), and that $\norm{g_j}_\infty \leq 2^{O(j)}\varepsilon\norm{f}_\infty$. With this in mind, we can expand out $\norm{T_\rho f}$ by the Bottom-Up Decomposition and apply \Cref{thm:hypercontraction} to get:
\begin{align*}
    \norm{T_\rho f}_4 &\leq \sum\limits_{j=0}^i \norm{T_\rho \fup{j}}_4\\
    &\leq \sum\limits_{j=0}^i \rho^i\norm{\fup{j}}_4 + \norm{\text{err}_j}_4\\
    &\leq \frac{1}{2}\varepsilon^{1/4}\norm{f}_2^{1/2}\norm{f}_\infty^{1/2} + \sum\limits_{j=0}^i \norm{\text{err}_j}_4\\
    &\leq  \frac{1}{2}\varepsilon^{1/4}\norm{f}_2^{1/2}\norm{f}_\infty^{1/2} + \sum\limits_{j=0}^i \norm{\text{err}_j}^{1/2}_2\norm{\text{err}_j}_\infty^{1/2}\\
    &\leq \varepsilon^{1/4}\norm{f}_2^{1/2}\norm{f}_\infty^{1/2}
\end{align*}
where we have assumed that $\rho$ is a sufficiently small constant. Taking the fourth power of both sides completes the proof.
\end{proof}

\bibliographystyle{amsalpha}  
\bibliography{references} 

\end{document}